\documentclass[draftclsnofoot,onecolumn,12pt]{IEEEtran}

\usepackage{amsmath,amssymb}
\usepackage{braket}
\usepackage{xcolor}
\usepackage{graphicx}
\usepackage{siunitx}
\usepackage{cite}
\usepackage{tikz}
\usepackage{booktabs}
\usepackage{tabularx}
\usepackage{array}
\usetikzlibrary{arrows.meta,positioning,calc,fit,shapes.geometric}

\newcommand{\Heff}{H_{\mathrm{eff}}}
\newcommand{\Hzero}{H_{0}}
\newcommand{\Hplus}{H_{+}}
\newcommand{\Hminus}{H_{-}}
\newcommand{\Hatom}{H_{\mathrm{atom}}}

\newcommand{\SNR}{\mathrm{SNR}}
\newcommand{\RMSE}{\mathrm{RMSE}}

\newcommand{\OmTwoThree}{\Omega_{23}}
\newcommand{\OmTwoFour}{\Omega_{24}}
\newcommand{\OmThreeFour}{\Omega_{34}}

\newtheorem{remark}{Remark}
\newtheorem{proposition}{Proposition}

\begin{document}

\title{LO-Free Phase and Amplitude Recovery of an RF Signal with a DC-Stark-Enabled Rydberg Receiver}

\author{Vladislav Katkov and Nikola Zlatanov
\thanks{The authors are with Innopolis University, Innopolis, 420500, Russia (e-mails: \texttt{v.katkov@innopolis.university; n.zlatanov@innopolis.ru}).}%
}

\maketitle

\begin{abstract}
We present a theoretical framework for recovering the amplitude and carrier phase of a single received RF field with a Rydberg-atom receiver, without injecting an RF local oscillator (LO) into the atoms. The key enabling mechanism is a static DC bias applied to the vapor cell: by Stark-mixing a near-degenerate Rydberg pair, the bias activates an otherwise absent upper optical pathway and closes a phase-sensitive loop within a receiver driven only by the standard probe/coupling pair and the received RF field. For a spatially uniform bias, we derive an effective four-level rotating-frame Hamiltonian of Floquet form and show that the periodic steady state obeys an exact harmonic phase law, so that the $n$th probe harmonic carries the factor $e^{in\Phi_S}$. This yields direct estimators for the signal phase and amplitude from a demodulated probe harmonic, with amplitude recovery obtained by inverting an injective harmonic response map. In the high-SNR regime, we derive explicit RMSE laws and use them to identify distinct phase-optimal and amplitude-optimal bias-controlled mixing angles, together with a weighted joint-design criterion and a balanced compromise angle that equalizes the fractional phase and amplitude penalties. We then extend the analysis to nonuniform DC bias through quasistatic spatial averaging and show that bias inhomogeneity reduces coherent gain for phase readout while also reshaping the amplitude-response slope. Numerical examples validate the phase law, illustrate response-map inversion and mixing-angle trade-offs, and quantify the penalties induced by bias nonuniformity. The results establish a minimal route to coherent Rydberg reception of a single RF signal without an auxiliary RF LO in the atoms.
\end{abstract}

\section{Introduction}

Rydberg-atom electrometry based on ladder electromagnetically induced transparency (EIT) provides a well-established route to measuring radio-frequency electric fields in vapor cells and has developed into a versatile platform for atom-based receivers and demodulators \cite{Sedlacek2012NatPhys,Simons2021MeasSens}. In many receiver applications, however, the objective is not merely field magnitude but the amplitude and carrier phase of a single received RF field. In this paper the received signal is
\begin{equation}
\mathbf{E}_S(t)=A_S\cos(\omega_S t+\Phi_S)\,\mathbf{e}_S,
\label{eq:signal_field_intro}
\end{equation}
where $A_S$ and $\Phi_S$ are unknown, while $\omega_S$ is assumed known or separately tracked for synchronous demodulation. Recovering both quantities is attractive for coherent communications, interferometric ranging, and direction-finding architectures. Standard EIT and Autler--Townes readout, however, are primarily amplitude sensitive and do not by themselves supply an internal phase reference for a single unknown RF tone.

Existing Rydberg receiver work spans a broader landscape than phase-sensitive schemes alone. On the communications and demodulation side, room-temperature Rydberg sensors have already been used for digital communication with amplitude-modulated microwave fields \cite{Meyer2018APL_DigitalComm}, continuously tunable RF carriers \cite{Song2019OEx_TunableCarrier}, amplitude-modulated baseband reception \cite{Jiao2019APEx_AMBaseband}, AM/FM radio reception \cite{Anderson2021TAP_AMFM}, multi-band and stereo reception \cite{Holloway2021APMag_Multiband}, and simultaneous multiband demodulation \cite{Meyer2023PRApplied_Multiband}. More recent work extends this communications agenda to quadrature-amplitude modulation \cite{Nowosielski2024OEx_QAM}, millimeter-wave atomic reception \cite{Legaie2024AVS_MMWReceiver}, and continuous-broadband reception based on AC-Stark tuning and Floquet sidebands \cite{Song2024APL_ContinuousBroadband}. These papers establish the breadth of atom-based reception, but they do not address the problem considered here: LO-free recovery of the carrier phase and amplitude of a single received RF field using an internally closed phase-sensitive loop for which the received field itself is the only time-varying RF tone driving the atoms.

Within the narrower class of phase-sensitive Rydberg schemes, the required reference is usually supplied in one of two ways. One class injects an RF local oscillator (LO) into the atoms, as in atomic-mixer phase measurement \cite{Simons2019APL_Mixer}, phase-modulated-signal reception \cite{Holloway2019LAWP_PhaseMod}, antenna-integrated phase- and amplitude-resolved reception \cite{Simons2019IEEEAccess_Antenna}, superheterodyne sensing \cite{Jing2020NatPhys_Superhet}, and high-sensitivity phase-modulation reception for frequency-division-multiplexing communication \cite{Cai2023PRApplied_PMReceiver}. Another class avoids an RF LO in the atoms but introduces additional time-varying optical or electrical reference structure, including interacting-dark-state phase proposals \cite{Lin2022AO_DarkStatePhase}, internal-state interferometry \cite{Anderson2022PRApplied_Interferometer}, closed-loop quantum interferometry \cite{Berweger2023PRApplied_ClosedLoop}, all-optical phase detection \cite{Schmidt2025PRL_AllOpticalPhase}, and optically biased hybrid nonlinear interferometry \cite{Borowka2025NatCommun_OpticalBias}. The general fractured-loop framework of Kasza \emph{et al.} \cite{Kasza2025PRA_FracturedLoops} clarifies how such multi-tone receivers can be modeled, but it does not identify a static DC field as the sole loop-closing resource. The distinguishing feature of the present work is therefore not phase sensitivity alone, but the specific way it is achieved: a \emph{static} DC bias Stark-mixes a near-degenerate Rydberg pair, activates the missing upper optical leg, and enables joint carrier-phase and amplitude recovery without an RF LO in the atoms and without any additional time-varying optical or RF reference that must be phase-related to the received signal inside the atomic medium.

This paper develops such a route. The central idea is to use a static DC bias field to Stark-mix a near-degenerate upper Rydberg pair. In the dressed basis, that mixing activates an otherwise absent upper optical pathway and closes a phase-sensitive loop, while the receiver as a whole is still driven only by the probe field, the coupling field, and the received RF signal itself. In a co-rotating description, the loop carries the phase combination $(\omega_S t+\Phi_S)$, so the long-time state becomes periodic at $\omega_S$. The probe harmonics therefore acquire an exact phase dependence $e^{in\Phi_S}$, which makes the carrier phase accessible from the argument of a demodulated harmonic, while the signal amplitude is recovered by inverting a harmonic response map. Figure~\ref{fig:intro_concept} summarizes the mechanism. Throughout the paper, ``LO-free'' means that no auxiliary \emph{RF} local oscillator is injected into the atoms; ordinary optical/electronic synchronous demodulation outside the vapor cell is still used to extract the complex probe harmonic at the known carrier frequency.

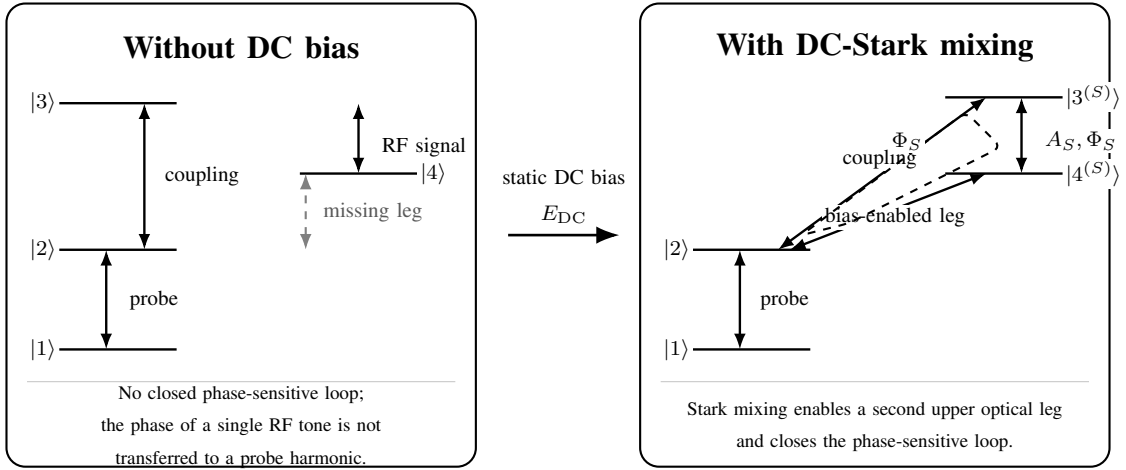
\begin{figure*}[t]
\centering
\begin{tikzpicture}[
    x=0.78cm,y=0.78cm,
    level/.style={line width=0.9pt},
    trans/.style={{Latex[length=2.1mm]}-{Latex[length=2.1mm]}, line width=0.9pt},
    disabled/.style={{Latex[length=2.1mm]}-{Latex[length=2.1mm]}, dashed, gray, line width=0.8pt},
    panel/.style={draw, rounded corners=0.25cm, line width=0.8pt},
    ttl/.style={font=\normalsize\bfseries},
    lab/.style={font=\scriptsize, fill=white, inner sep=1pt},
    softlab/.style={font=\scriptsize, text=gray!70!black, fill=white, inner sep=1pt},
    note/.style={font=\fontsize{7.2}{8.2}\selectfont, align=center}
]

\draw[panel] (0,0) rectangle (8.0,7.9);
\node[ttl] at (4.0,7.15) {Without DC bias};

\draw[gray!45] (0.35,1.45) -- (7.65,1.45);

\draw[level] (0.9,2.0) -- (2.9,2.0);
\draw[level] (0.9,3.7) -- (2.9,3.7);
\draw[level] (0.9,6.2) -- (2.9,6.2);
\draw[level] (5.0,5.0) -- (7.0,5.0);

\node[anchor=east,lab] at (0.9,2.0) {$\ket{1}$};
\node[anchor=east,lab] at (0.9,3.7) {$\ket{2}$};
\node[anchor=east,lab] at (0.9,6.2) {$\ket{3}$};
\node[anchor=west,lab] at (7.0,5.0) {$\ket{4}$};

\draw[trans] (1.7,2.0) -- (1.7,3.7);
\draw[trans] (2.35,3.7) -- (2.35,6.2);
\draw[trans] (6.0,5.0) -- (6.0,6.2);
\draw[disabled] (5.1,3.7) -- (5.1,5.0);

\node[lab,anchor=west] at (2.05,2.85) {probe};
\node[lab,anchor=west] at (2.65,4.95) {coupling};
\node[lab,anchor=west] at (6.35,5.45) {RF signal};
\node[softlab,anchor=west] at (5.35,4.35) {missing leg};

\node[note,text width=7.0cm] at (4.0,0.68)
{No closed phase-sensitive loop;\\
the phase of a single RF tone is not transferred to a probe harmonic.};

\draw[-{Latex[length=3.1mm]}, line width=1.0pt]
    (8.55,3.95) -- (10.45,3.95)
    node[midway,above=3pt,lab,align=center]
    {static DC bias\\$E_{\mathrm{DC}}$};

\draw[panel] (10.8,0) rectangle (18.8,7.9);
\node[ttl] at (14.8,7.15) {With DC-Stark mixing};

\draw[gray!45] (11.15,1.45) -- (18.45,1.45);

\draw[level] (11.7,2.0) -- (13.7,2.0);
\draw[level] (11.7,3.7) -- (13.7,3.7);
\draw[level] (16.0,6.3) -- (18.0,6.3);
\draw[level] (16.0,5.0) -- (18.0,5.0);

\node[anchor=east,lab] at (11.7,2.0) {$\ket{1}$};
\node[anchor=east,lab] at (11.7,3.7) {$\ket{2}$};
\node[anchor=west,lab] at (18.0,6.3) {$\ket{3^{(S)}}$};
\node[anchor=west,lab] at (18.0,5.0) {$\ket{4^{(S)}}$};

\draw[trans] (12.5,2.0) -- (12.5,3.7);
\draw[trans] (13.15,3.7) -- (16.7,6.3);
\draw[trans] (13.35,3.7) -- (16.7,5.0);
\draw[trans] (17.3,5.0) -- (17.3,6.3);

\node[lab,anchor=west] at (12.8,2.85) {probe};
\node[lab,anchor=west] at (14.2,5.25) {coupling};
\node[lab,anchor=west] at (13.9,4.25) {bias-enabled leg};
\node[lab,anchor=west] at (17.65,5.55) {$A_S,\Phi_S$};

\draw[dashed, rounded corners=3pt, line width=0.8pt]
    (13.55,3.95) -- (16.35,6.05) -- (16.95,5.45) -- (14.35,4.15) -- cycle;
\node[lab] at (15.35,5.55) {$\Phi_S$};

\node[note,text width=7.0cm] at (14.8,0.68)
{Stark mixing enables a second upper optical leg\\
and closes the phase-sensitive loop.};

\end{tikzpicture}
\caption{Conceptual mechanism of the proposed receiver. Left: without a DC bias, the standard ladder-EIT configuration plus the received RF coupling does not supply the additional upper optical pathway needed for a phase-sensitive loop. Right: a static DC bias Stark-mixes the upper pair, activates a bias-enabled optical leg, and closes a loop within a receiver driven only by the probe, coupling, and received RF fields. The resulting periodic probe response carries the phase of the received signal in its harmonics.}
\label{fig:intro_concept}
\end{figure*}

The paper makes four concrete contributions:
\begin{enumerate}
\item It introduces an LO-free reception mechanism for a single RF signal in which a static DC bias closes the phase-sensitive loop through Stark mixing.
\item It derives a reduced four-level rotating-frame model of Floquet form and establishes, under the uniform-bias benchmark, an exact harmonic phase law for the periodic steady state.
\item It develops phase and amplitude estimators from a demodulated probe harmonic, with amplitude recovery posed as inversion of an injective harmonic response map.
\item It derives high-SNR RMSE laws, identifies phase-optimal, amplitude-optimal, weighted joint-design, and balanced-compromise mixing angles, and quantifies the impact of DC-bias nonuniformity through spatial averaging.
\end{enumerate}

The treatment in this paper is theoretical: we first analyze the spatially uniform DC-bias benchmark, for which the reduced periodically driven model is exact and the estimation problem is posed most cleanly, and then treat bias nonuniformity as a quasistatic spatial-average perturbation of that benchmark. 

The remainder of the paper is organized as follows. Section~\ref{sec:model} introduces the reduced model and the signal-estimation problem; Section~\ref{sec:hamiltonian} derives the rotating-frame Hamiltonian; Section~\ref{sec:pss} establishes the periodic steady state and harmonic phase law; Sections~\ref{sec:estimators} and \ref{sec:performance} develop the estimators and high-SNR performance laws; Section~\ref{sec:bias_cases} treats uniform and nonuniform bias; Section~\ref{sec:numerics} presents numerical examples; and Sections~\ref{sec:discussion} and \ref{sec:conclusion} discuss scope and conclusions. Technical derivations are collected in the appendices. Species-specific level selection and full-manifold optimization are left for future implementation work.

\section{Reduced receiver model and estimation target}
\label{sec:model}

This section introduces the minimal reduced receiver model used throughout the paper and states the corresponding estimation target. The main purpose is to define the field and basis conventions, to distinguish clearly between the physical RF amplitude $A_S$, the underlying signal Rabi scale $\Omega_S$, and the effective dressed-basis coupling $\OmThreeFour$, and to formulate the problem first under a spatially uniform DC-bias benchmark.

\subsection{Fields, basis states, and conventions}
\label{subsec:fields_basis}

We consider a ladder-EIT vapor-cell receiver driven by a static bias field, a probe field, a coupling field, and the received RF signal
\begin{align}
\mathbf{E}_{\mathrm{DC}}(\mathbf{r}) &= E_z(\mathbf{r})\,\mathbf{e}_z, \nonumber\\
\mathbf{E}_{p}(t) &= A_p\cos(\omega_p t)\,\mathbf{e}_p, \nonumber\\
\mathbf{E}_{c}(t) &= A_c\cos(\omega_c t)\,\mathbf{e}_c, \nonumber\\
\mathbf{E}_{S}(t) &= A_S\cos(\omega_S t+\Phi_S)\,\mathbf{e}_S.
\label{eq:fields}
\end{align}
The probe and coupling amplitudes and frequencies $(A_p,\omega_p)$ and $(A_c,\omega_c)$ are assumed known, and the carrier frequency $\omega_S$ is assumed known or separately tracked for synchronous demodulation. The unknown quantities to be estimated are the received-signal amplitude $A_S$ and carrier phase $\Phi_S$.

The constructive model starts from four field-free states $\ket{1}$, $\ket{2}$, $\ket{3}$, and $\ket{4}$. Before the DC bias is applied, the probe laser addresses $\ket{1}\leftrightarrow\ket{2}$, the coupling laser addresses $\ket{2}\leftrightarrow\ket{3}$, and the upper bare pair $\{\ket{3},\ket{4}\}$ is the pair that will later be mixed by the static bias field. At fixed bias, diagonalization of that upper $2\times2$ block produces two Stark states, denoted by $\ket{3^{(S)}}$ and $\ket{4^{(S)}}$. The reduced receiver basis used in the main text is therefore
\[
\{\ket{1},\ket{2},\ket{3^{(S)}},\ket{4^{(S)}}\}.
\]
Throughout the paper, $\ket{3}$ and $\ket{4}$ refer to the bare upper states before bias mixing, whereas $\ket{3^{(S)}}$ and $\ket{4^{(S)}}$ refer to the corresponding Stark states after mixing.

Using the polarization-projected dipole matrix elements
\[
\mu_{12}^{p}=\bra{1}\hat{\boldsymbol{\mu}}\!\cdot\!\mathbf{e}_p\ket{2},\qquad
\mu_{23}^{c}=\bra{2}\hat{\boldsymbol{\mu}}\!\cdot\!\mathbf{e}_c\ket{3},\qquad
\mu_{34}^{S}=\bra{3}\hat{\boldsymbol{\mu}}\!\cdot\!\mathbf{e}_S\ket{4},
\]

\begin{equation}
\Omega_p=\frac{|\mu_{12}^{p}|A_p}{2\hbar},\qquad
\Omega_c=\frac{|\mu_{23}^{c}|A_c}{2\hbar},\qquad
\Omega_S=\frac{|\mu_{34}^{S}|A_S}{2\hbar}.
\label{eq:rabi_scales}
\end{equation}
The factor $1/2$ appears because the real fields in (\ref{eq:fields}) are written as cosines, while only the positive-frequency components enter the rotating-wave Hamiltonian.

The quantity $\Omega_S$ is the underlying signal-strength parameter associated directly with the received RF amplitude $A_S$. After Stark mixing, however, the coupling that enters the reduced Hamiltonian is not $\Omega_S$ itself but an effective upper-state coupling $\OmThreeFour$, introduced below. We formulate the estimation problem in terms of $(\Omega_S,\Phi_S)$ because $\Omega_S$ remains directly proportional to the physical field amplitude $A_S$ through (\ref{eq:rabi_scales}).

\subsection{DC-Stark-induced effective couplings}
\label{subsec:minimal_realization}

We use a minimal two-state Stark model in which the DC bias mixes only the near-degenerate upper bare pair $\{\ket{3},\ket{4}\}$, while the lower states $\ket{1}$ and $\ket{2}$ are left unchanged. This model is intentionally minimal: its purpose is to expose the loop-closure mechanism analytically, not to provide a species-specific manifold calculation. The formal diagonalization of the Stark block is given in Section~\ref{sec:hamiltonian}; here we introduce the resulting effective couplings.

If the coupling laser drives only the bare transition $\ket{2}\leftrightarrow\ket{3}$ before mixing, then the same optical tone drives two transitions in the Stark basis:
\begin{equation}
\OmTwoThree=\Omega_c\cos\theta,\qquad
\OmTwoFour=\Omega_c\sin\theta,
\label{eq:effective_couplings}
\end{equation}
where $\theta$ is the bias-controlled Stark-mixing angle. The second coupling $\ket{2}\leftrightarrow\ket{4^{(S)}}$ is absent when $\theta=0$ and is the bias-enabled optical leg that closes the phase-sensitive loop.

The received RF signal couples the Stark states with effective strength
\begin{equation}
\OmThreeFour=\Omega_S\cos(2\theta)
\label{eq:effective_signal_coupling}
\end{equation}
in the real-coupling convention used throughout the paper. Constant dipole phases and overall signs are absorbed into the basis conventions, so that $\OmTwoThree$, $\OmTwoFour$, and $\OmThreeFour$ may be taken real in the minimal model. The hierarchy of signal-strength quantities is therefore
\begin{equation}
A_S
\;\longrightarrow\;
\Omega_S=\frac{|\mu_{34}^{S}|A_S}{2\hbar}
\;\longrightarrow\;
\OmThreeFour=\Omega_S\cos(2\theta).
\label{eq:signal_hierarchy}
\end{equation}
The first arrow converts the physical RF amplitude into the underlying Rabi scale associated with the received field; the second arrow incorporates DC-Stark mixing and yields the effective dressed-basis coupling that appears in the Hamiltonian.

These relations also make the central design trade-off visible already at the model level. Increasing $\theta$ turns on the bias-enabled optical leg $\OmTwoFour$, which is required for loop closure, but it simultaneously reduces the effective RF coupling $\OmThreeFour$ through the factor $\cos(2\theta)$. This competition is the origin of the distinct phase- and amplitude-optimal mixing angles derived later.

The general theory developed below does not require the specific trigonometric forms in (\ref{eq:effective_couplings}) and (\ref{eq:effective_signal_coupling}); it requires only that the reduced Stark-basis description be characterized by well-defined effective couplings $\OmTwoThree$, $\OmTwoFour$, and $\OmThreeFour$. The trigonometric relations belong to the minimal constructive model and are used mainly to expose the operating-point trade-offs controlled by the DC bias.

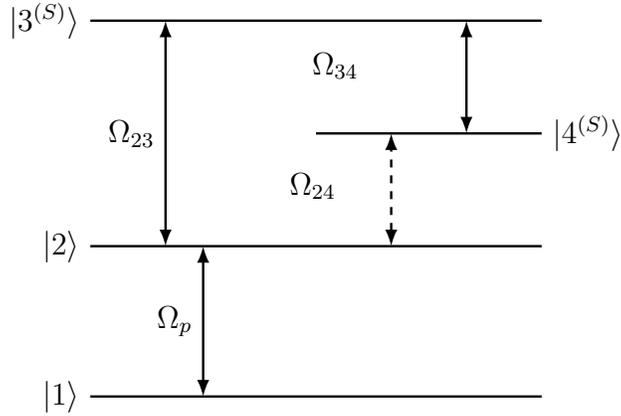
\begin{figure}[t]
\centering
\begin{tikzpicture}[
    level/.style={line width=0.9pt},
    trans/.style={{Latex[length=2.2mm]}-{Latex[length=2.2mm]}, line width=0.9pt},
    dashedtrans/.style={{Latex[length=2.2mm]}-{Latex[length=2.2mm]}, dashed, line width=0.9pt},
]

\draw[level] (0,0) -- (6,0);
\draw[level] (0,2) -- (6,2);
\draw[level] (0,5) -- (6,5);
\draw[level] (3,3.5) -- (6,3.5);

\node[anchor=east] at (0,0) {$\ket{1}$};
\node[anchor=east] at (0,2) {$\ket{2}$};
\node[anchor=east] at (0,5) {$\ket{3^{(S)}}$};
\node[anchor=west] at (6,3.5) {$\ket{4^{(S)}}$};

\draw[trans]       (1.5,0.0) -- (1.5,2.0);
\draw[trans]       (1.0,2.0) -- (1.0,5);
\draw[dashedtrans] (4,2.0) -- (4,3.5);
\draw[trans]       (5,3.5) -- (5,5);

\node[anchor=east] at (1.5,1.0) {$\Omega_p$};
\node[anchor=east] at (1,3.5) {$\OmTwoThree$};
\node[anchor=west] at (2.5,2.8) {$\OmTwoFour$};
\node[anchor=west] at (2.8,4.4) {$\OmThreeFour$};

\end{tikzpicture}
\caption{Reduced receiver model in the Stark basis. A DC bias mixes the upper bare pair so that the coupling tone drives both $\ket{2}\leftrightarrow\ket{3^{(S)}}$ and the bias-enabled leg $\ket{2}\leftrightarrow\ket{4^{(S)}}$, while the received RF field couples the Stark states. In the minimal constructive model, $\OmTwoThree=\Omega_c\cos\theta$, $\OmTwoFour=\Omega_c\sin\theta$, and $\OmThreeFour=\Omega_S\cos(2\theta)$.}
\label{fig:level_scheme}
\end{figure}

In the co-rotating description derived in Section~\ref{sec:hamiltonian}, a gauge choice places the explicit factor $e^{\pm i(\omega_S t+\Phi_S)}$ on the bias-enabled optical leg rather than on the physical RF-driven upper transition. This is only a representation choice: the received RF field still physically couples $\ket{3^{(S)}}$ and $\ket{4^{(S)}}$, and only the total loop phase is observable.

\subsection{Uniform-bias benchmark and estimation target}
\label{subsec:problem_statement}

The main theory in Sections~\ref{sec:hamiltonian}--\ref{sec:performance} is developed for the spatially uniform-bias benchmark
\begin{equation}
E_z(\mathbf{r})\equiv E_z.
\label{eq:uniform_bias}
\end{equation}
In the weak-probe regime, the optical readout is governed by the probe coherence
\[
\rho_{21}(t,\Phi_S)\equiv \bra{2}\rho(t,\Phi_S)\ket{1},
\]
so its harmonic content is the central estimation observable. Section~\ref{sec:estimators} later relates $\rho_{21}$ directly to the measured probe response.

Under the uniform-bias assumption, the reduced Hamiltonian derived in Section~\ref{sec:hamiltonian} has the Floquet form
\begin{equation}
\Heff(t,\Phi_S)=\Hzero+\Hplus e^{i(\omega_S t+\Phi_S)}+\Hminus e^{-i(\omega_S t+\Phi_S)},
\label{eq:floquet_form}
\end{equation}
and the probe coherence admits the harmonic expansion
\begin{equation}
\rho_{21}(t,\Phi_S)=\sum_{n\in\mathbb{Z}} \rho_{21}^{(n)}(\Phi_S)\,e^{in\omega_S t}.
\label{eq:rho21_harmonics}
\end{equation}
The estimation problem is therefore to recover $(A_S,\Phi_S)$, or equivalently $(\Omega_S,\Phi_S)$ through (\ref{eq:rabi_scales}), from one or more harmonics of $\rho_{21}(t,\Phi_S)$. We formulate amplitude recovery in terms of $\Omega_S$ rather than $\OmThreeFour$ because $\Omega_S$ is the quantity that maps directly back to the physical RF amplitude $A_S$.

Section~\ref{sec:bias_cases} relaxes (\ref{eq:uniform_bias}) and treats the more general case $E_z(\mathbf{r})=E_0+\delta E(\mathbf{r})$ as a quasistatic spatial average of local uniform-bias responses.

\section{Rotating-frame Hamiltonian under uniform bias}
\label{sec:hamiltonian}

We now convert the conceptual reduced receiver model into the periodically driven Hamiltonian used in the Floquet analysis. The derivation assumes the uniform-bias benchmark in (\ref{eq:uniform_bias}) and the minimal two-state Stark model in which the upper pair $\{\ket{3},\ket{4}\}$ is near-degenerate and sufficiently isolated that mixing with additional nearby Rydberg states can be neglected at leading order.

\subsection{DC Stark mixing of the upper bare pair}
\label{subsec:stark_mixing}

In the bare basis $\{\ket{1},\ket{2},\ket{3},\ket{4}\}$ the field-free Hamiltonian is
\begin{equation}
\Hatom=\hbar\sum_{n=1}^{4}\omega_n\ket{n}\!\bra{n}.
\label{eq:Hatom}
\end{equation}
A static bias field $\mathbf{E}_{\mathrm{DC}}=E_z\,\mathbf{e}_z$ adds the Stark interaction
\begin{equation}
H_0'=\Hatom-\hat{\mu}_z E_z,\qquad \hat{\mu}_z=\hat{\boldsymbol{\mu}}\!\cdot\!\mathbf{e}_z.
\label{eq:stark_hamiltonian}
\end{equation}
For field-free eigenstates of definite parity, the diagonal matrix elements vanish, $\bra{n}\hat{\mu}_z\ket{n}=0$ \cite{Gallagher1994Rydberg}. In the minimal model, only the upper bare pair is mixed, so the relevant Stark block is
\begin{equation}
H'_{34}=
\begin{bmatrix}
\hbar\omega_3 & -E_z\mu_{34}^{z} \\
-E_z\mu_{43}^{z} & \hbar\omega_4
\end{bmatrix},
\qquad
\mu_{43}^{z}=\mu_{34}^{z*}.
\label{eq:H34_block}
\end{equation}
Choosing phases so that $\mu_{34}^{z}$ is real and positive, diagonalization yields the Stark states
\begin{align}
\ket{3^{(S)}} &= \cos\theta\,\ket{3}+\sin\theta\,\ket{4}, \nonumber\\
\ket{4^{(S)}} &= -\sin\theta\,\ket{3}+\cos\theta\,\ket{4},
\label{eq:stark_states}
\end{align}
with bias-controlled Stark-mixing angle
\begin{equation}
\tan(2\theta)=\beta,\qquad
\beta\equiv \frac{2E_z|\mu_{34}^{z}|}{\hbar\Delta_{34}},\qquad
\Delta_{34}\equiv \omega_3-\omega_4>0,
\label{eq:theta_def}
\end{equation}
where $0\le \theta\le \pi/4$. The limit $\theta=0$ corresponds to no mixing, while $\theta\to\pi/4$ corresponds to strong mixing within the isolated-pair model. The dressed splitting becomes
\begin{equation}
\omega_{34}^{(S)}\equiv \omega_3^{(S)}-\omega_4^{(S)}
=
\sqrt{(\omega_3-\omega_4)^2+\left(\frac{2E_z|\mu_{34}^{z}|}{\hbar}\right)^2}.
\label{eq:dressed_splitting}
\end{equation}

\subsection{Effective Hamiltonian in the co-rotating frame}
\label{subsec:rot_frame}

Appendix~\ref{app:stark_derivation} gives the full derivation from the bare-state Hamiltonian, including the rotating-wave approximation and the change of frame. Here we record the resulting reduced Hamiltonian in the Stark basis $\{\ket{1},\ket{2},\ket{3^{(S)}},\ket{4^{(S)}}\}$. From this point onward, $\Heff$, $\Hzero$, $\Hplus$, and $\Hminus$ are written in angular-frequency units, so that the physical Hamiltonian is $\hbar\Heff$.

The reduced Hamiltonian can be written as
\begin{equation}
\Heff(t,\Phi_S)=\Hzero+\Hplus e^{i(\omega_S t+\Phi_S)}+\Hminus e^{-i(\omega_S t+\Phi_S)}.
\label{eq:Heff_main}
\end{equation}
The time-independent block is
\begin{equation}
\Hzero=
\begin{bmatrix}
0 & \Omega_p & 0 & 0\\
\Omega_p & -\Delta_p & \OmTwoThree & 0\\
0 & \OmTwoThree & -(\Delta_p+\Delta_c) & \OmThreeFour\\
0 & 0 & \OmThreeFour & -(\Delta_p+\Delta_c-\Delta_S)
\end{bmatrix},
\label{eq:H0}
\end{equation}
and the $\pm1$ Fourier blocks are
\begin{equation}
\Hplus=
\begin{bmatrix}
0&0&0&0\\
0&0&0&\OmTwoFour\\
0&0&0&0\\
0&0&0&0
\end{bmatrix},
\qquad
\Hminus=\Hplus^\dagger.
\label{eq:Hpm}
\end{equation}
The detunings are
\begin{equation}
\Delta_p=\omega_p-(\omega_2-\omega_1),\qquad
\Delta_c=\omega_c-(\omega_3^{(S)}-\omega_2),\qquad
\Delta_S=\omega_S-\omega_{34}^{(S)}.
\label{eq:detunings}
\end{equation}

The structure of (\ref{eq:H0})--(\ref{eq:Hpm}) is physically transparent. The static block $\Hzero$ contains the probe coupling $\Omega_p$, the primary upper optical coupling $\OmTwoThree$, and the effective RF-driven Stark-state coupling $\OmThreeFour$. The periodic blocks $\Hplus$ and $\Hminus$ contain the bias-enabled optical leg $\OmTwoFour$ and are the only source of explicit time periodicity. Because the Hamiltonian contains only Fourier components at $0$ and $\pm1$, the harmonic-balance equations developed in Section~\ref{sec:pss} will couple only neighboring harmonics.

In the constructive model, one substitutes (\ref{eq:effective_couplings}) and (\ref{eq:effective_signal_coupling}) into (\ref{eq:H0})--(\ref{eq:Hpm}). Although the received RF field physically drives the Stark-state transition $\ket{3^{(S)}}\leftrightarrow\ket{4^{(S)}}$, the chosen rotating frame transfers the explicit phase factor $e^{\pm i(\omega_S t+\Phi_S)}$ to the $\ket{2}\leftrightarrow\ket{4^{(S)}}$ matrix element. This does not change the physics: only the total loop phase is gauge invariant. In the representation used here, the unknown signal amplitude enters through $\OmThreeFour$, while the unknown signal phase enters only through the combination $(\omega_S t+\Phi_S)$.

\subsection{Validity and scope of the reduced Hamiltonian}

The isolated-pair Stark model is the minimal construction needed to expose the DC-Stark-enabled loop-closure mechanism. It neglects additional nearby Stark states and bias-induced mixing with lower-lying ladder states, whose much larger energy separations make their Stark admixture subleading in the minimal model. It also neglects magnetic-sublevel structure (degenerate here in the absence of an applied magnetic field), Doppler averaging, transport, and diagonal time-periodic terms that arise in fuller dressed-basis descriptions and are discarded here as an additional simplifying approximation of the minimal model. Such effects can shift quantitative response maps, resonance conditions, and optimal operating points.

These omissions do not, however, alter the main structural point of the reduced model: the received signal creates a periodically driven closed loop whose phase enters through $(\omega_S t+\Phi_S)$. The harmonic phase law derived in Section~\ref{sec:pss} is therefore robust to many model refinements, provided the dependence on the unknown signal phase remains only through that combination and the dissipator remains time independent. The present four-level construction should thus be viewed as the minimal analytical model that isolates the reception principle cleanly, while more detailed manifold-level descriptions are deferred to future implementation work.


\section{Periodic steady state and exact harmonic phase law}
\label{sec:pss}
The reduced Hamiltonian derived in Section~\ref{sec:hamiltonian} is periodic at the received-signal frequency $\omega_S$. Consequently, the long-time response of the driven open system is not time independent in general, but periodic with the same fundamental period. The goal of this section is to show that the unknown carrier phase enters that periodic response in an especially simple way: each harmonic acquires the exact phase factor $e^{in\Phi_S}$.

\subsection{Driven open-system model}
\label{subsec:master_equation}

Under continuous periodic driving, the density operator synchronizes to the drive rather than converging to a static steady state. We model the receiver with the Lindblad master equation
\begin{equation}
\frac{d\rho}{dt}
=
-i[\Heff(t,\Phi_S),\rho]+\mathcal{D}[\rho],
\label{eq:master_eq}
\end{equation}
where
\[
\mathcal{D}[\rho]
=
\sum_k\left(L_k\rho L_k^\dagger-\frac{1}{2}\{L_k^\dagger L_k,\rho\}\right)
\]
is the time-independent Lindblad dissipator, and the $L_k$ are jump operators describing radiative decay and dephasing.

Because $\Heff(t,\Phi_S)$ is periodic at $\omega_S$, the long-time state is generally a periodic steady state (PSS) with period
\[
T=\frac{2\pi}{\omega_S}.
\]

\subsection{Exact harmonic phase law}
\label{subsec:phase_factorization}

The key observation is that the unknown signal phase appears in the Hamiltonian only through the combination $(\omega_S t+\Phi_S)$. Changing $\Phi_S$ is therefore equivalent to shifting the periodic drive in time. Write the periodic steady state as the Fourier series:
\begin{equation}
\rho(t,\Phi_S)=\sum_{n\in\mathbb{Z}} \rho^{(n)}(\Phi_S)\,e^{in\omega_S t}.
\label{eq:rho_fourier}
\end{equation}

\begin{proposition}[Exact harmonic phase law]
Assume that:
(i) the Hamiltonian depends on the unknown signal phase only through the factors $e^{\pm i(\omega_S t+\Phi_S)}$ in (\ref{eq:Heff_main}),
(ii) the dissipator $\mathcal{D}[\rho]$ in (\ref{eq:master_eq}) is time independent, and
(iii) the driven system admits a unique periodic steady state for each $\Phi_S$.
Then
\begin{equation}
\rho^{(n)}(\Phi_S)=P^{(n)}\,e^{in\Phi_S},
\qquad
P^{(n)}\equiv \rho^{(n)}(0),
\label{eq:factorization}
\end{equation}
for every integer $n$.
\end{proposition}

\noindent\textit{Proof.}
Because all phase dependence appears through $(\omega_S t+\Phi_S)$, one has
\[
\Heff(t,\Phi_S)=\Heff\!\left(t+\frac{\Phi_S}{\omega_S},0\right).
\]
Under the uniqueness assumption, the periodic steady states are related by a time shift:
\[
\rho_{\mathrm{PSS}}(t,\Phi_S)=\rho_{\mathrm{PSS}}\!\left(t+\frac{\Phi_S}{\omega_S},0\right).
\]
Expanding the time-shift relation in Fourier harmonics gives
\[
\rho^{(n)}(\Phi_S)=\rho^{(n)}(0)e^{in\Phi_S},
\]
which is precisely (\ref{eq:factorization}).

The matrices $P^{(n)}$ may therefore be interpreted as the \emph{reference-phase harmonic coefficients}, i.e., the Fourier coefficients evaluated at $\Phi_S=0$. For the probe coherence, (\ref{eq:factorization}) implies the immediate relations
\begin{equation}
|\rho_{21}^{(n)}(\Phi_S)|=|P_{21}^{(n)}|,
\qquad
\arg\!\big(\rho_{21}^{(n)}(\Phi_S)\big)
=
\arg\!\big(P_{21}^{(n)}\big)+n\Phi_S
\pmod{2\pi}.
\label{eq:phase_law_probe}
\end{equation}
Thus the harmonic magnitude is independent of $\Phi_S$, while the harmonic phase shifts linearly with slope $n$.

\begin{remark}
Detunings, linewidths, and dissipation modify the coefficients $P^{(n)}$, but not the exact factor $e^{in\Phi_S}$, provided the three assumptions above remain valid.
\end{remark}

\subsection{Harmonic-balance equations}
\label{subsec:harmonic_balance}

Substituting (\ref{eq:Heff_main}) and (\ref{eq:rho_fourier}) into (\ref{eq:master_eq}) and equating coefficients of $e^{in\omega_S t}$ gives
\begin{equation}
in\omega_S\,P^{(n)}
=
-i[\Hzero,P^{(n)}]
-i[\Hplus,P^{(n-1)}]
-i[\Hminus,P^{(n+1)}]
+\mathcal{D}[P^{(n)}],
\label{eq:harmonic_balance_eq}
\end{equation}
where $\mathcal{D}[\cdot]$ denotes the linear dissipative action of the time-independent Lindbladian. Because the Hamiltonian contains only Fourier components at $0$ and $\pm1$, each harmonic couples only to its nearest neighbors $(n-1,n,n+1)$.

Equation~(\ref{eq:harmonic_balance_eq}) shows that the full periodic-steady-state problem may be solved at $\Phi_S=0$ without loss of generality. Once the reference-phase coefficients $P^{(n)}$ have been computed, the response at arbitrary $\Phi_S$ follows immediately from (\ref{eq:factorization}). In practice, the infinite system is truncated to $n\in[-N,N]$ and solved numerically. Appendix~\ref{app:floquet_linear_system} gives a compact Liouville-space form of the truncated linear system.

The exact harmonic phase law reduces signal recovery to estimating the phase and magnitude of a demodulated probe harmonic. The next section turns that structure into explicit phase and amplitude estimators.

\section{Optical readout and recovery of signal phase and amplitude}
\label{sec:estimators}

This section converts the exact harmonic phase law into operational estimators based on the measured probe response. Throughout, ``LO-free'' means that no auxiliary \emph{RF} local oscillator is injected into the atoms; ordinary optical/electronic synchronous demodulation outside the vapor cell is still used to extract a complex probe harmonic at the known carrier frequency.

\subsection{Observable and calibrated harmonic readout}
\label{subsec:observable}

In the weak-probe regime, the susceptibility of the $\ket{1}\!\to\!\ket{2}$ transition is proportional to $\rho_{21}(t)$, so probe absorption and dispersion are proportional to $\Im[\rho_{21}(t)]$ and $\Re[\rho_{21}(t)]$, respectively \cite{Sedlacek2012NatPhys}. We therefore consider a complex readout channel of the form
\[
y(t)=\kappa\,\rho_{21}(t),
\]
where $\kappa$ is a fixed complex gain representing the optical/electronic detection chain. In practice, this constant gain can be absorbed into calibration, a reference measurement, or the response map used for amplitude inversion. For notational simplicity, we write the demodulated output directly in terms of the corresponding probe harmonic.

The $n$-th harmonic is obtained by synchronous demodulation over an integration interval spanning an integer number of received-signal periods:
\begin{equation}
\widehat{\rho}_{21}^{(n)}
=
\frac{1}{T_{\mathrm{int}}}\int_{t_0}^{t_0+T_{\mathrm{int}}} y(t)\,e^{-in\omega_S t}\,dt.
\label{eq:lockin_estimator}
\end{equation}
After calibration, $\widehat{\rho}_{21}^{(n)}$ is interpreted as the measured complex phasor associated with the $n$-th harmonic of the probe coherence.

\subsection{Phase recovery from a probe harmonic}
\label{subsec:phase_estimator}

For any nonzero harmonic, the exact phase law gives
\begin{equation}
\rho_{21}^{(n)}(\Phi_S)=P_{21}^{(n)}\,e^{in\Phi_S},
\qquad n\neq 0.
\label{eq:harmonic_phase}
\end{equation}
Equivalently,
\[
\arg\!\big(\rho_{21}^{(n)}(\Phi_S)\big)
=
\arg\!\big(P_{21}^{(n)}\big)+n\Phi_S
\pmod{2\pi}.
\]
Let $\rho_{21,\mathrm{ref}}^{(n)}$ be a nonzero reference harmonic measured or calibrated at known signal phase $\Phi_0$ under the same operating point and with the same readout-gain convention. Then the signal phase is estimated as
\begin{equation}
\widehat{\Phi}_S
=
\Phi_0+\frac{1}{n}\arg\!\left(\frac{\widehat{\rho}_{21}^{(n)}}{\rho_{21,\mathrm{ref}}^{(n)}}\right)
\pmod{\frac{2\pi}{n}}.
\label{eq:phase_estimator}
\end{equation}
The ratio removes any fixed complex gain and any fixed reference phase offset associated with the detection chain.

For $n=1$, (\ref{eq:phase_estimator}) recovers the phase modulo $2\pi$ and therefore avoids branch ambiguity. For $n>1$, the estimate is intrinsically defined only modulo $2\pi/n$, so practical use of higher harmonics requires prior phase-range information or an additional disambiguation step. Throughout the paper, the first harmonic is emphasized because it is typically the strongest phase-bearing harmonic and is branch-unambiguous.

\subsection{Amplitude recovery from an injective response map}
\label{subsec:amplitude_estimator}

Because the exact phase law separates the unknown carrier phase into the factor $e^{in\Phi_S}$, the magnitude of the $n$th harmonic depends only on the underlying signal strength $\Omega_S$. For a fixed operating point and calibrated readout convention, define the $n$th-harmonic response map
\begin{equation}
m_n(\Omega_S)\triangleq |P_{21}^{(n)}(\Omega_S)|.
\label{eq:response_map}
\end{equation}
We parameterize this map by $\Omega_S$, rather than by the dressed-basis coupling $\OmThreeFour$, because $\Omega_S$ is directly proportional to the physical RF amplitude $A_S$ through (\ref{eq:rabi_scales}).

If $m_n(\Omega_S)$ is injective on the chosen operating interval---for example, on a monotone branch---then a measured harmonic magnitude determines a unique $\Omega_S$, and the signal-strength estimate is
\begin{equation}
\widehat{\Omega}_S = m_n^{-1}\!\left(|\widehat{\rho}_{21}^{(n)}|\right).
\label{eq:amp_estimator}
\end{equation}
If the response map is not injective on the relevant interval, inversion becomes ambiguous and additional information, such as prior range knowledge or a second observable/harmonic, is required. In experiment, the response map may be calibrated directly in measured harmonic units, so absolute normalization of $\rho_{21}$ is not required as long as the same calibration convention is used during operation.

In the constructive model, the physical field amplitude follows from (\ref{eq:rabi_scales}) as
\begin{equation}
\widehat{A}_S=\frac{2\hbar}{|\mu_{34}^{S}|}\,\widehat{\Omega}_S.
\label{eq:A_from_Omega}
\end{equation}

For amplitude recovery, harmonic magnitude alone is not sufficient; the local slope of the response map also matters. We therefore define the logarithmic sensitivity
\begin{equation}
s_n(\Omega_S)
\triangleq
\frac{d\ln m_n(\Omega_S)}{d\ln\Omega_S}
=
\frac{\Omega_S}{m_n(\Omega_S)}\frac{dm_n(\Omega_S)}{d\Omega_S}.
\label{eq:log_sensitivity}
\end{equation}
The interpretation is straightforward: $|s_n|\approx 1$ corresponds to a locally proportional response, $|s_n|\ll 1$ indicates a flat or saturated region with poor amplitude sensitivity, and $s_n<0$ corresponds to a decreasing but still potentially invertible monotone branch.

Phase recovery therefore depends primarily on the angle of a demodulated harmonic phasor, whereas amplitude recovery depends on both the harmonic magnitude and the local slope of the response map. The next section quantifies this distinction through high-SNR RMSE laws.

\section{High-SNR error laws and mixing-angle design criteria}
\label{sec:performance}
This section quantifies the accuracy of phase and amplitude recovery under the spatially uniform-bias benchmark and uses the resulting local high-SNR laws to define operating-point criteria for the bias-controlled Stark-mixing angle $\theta$. The analysis is formulated for a noisy demodulated probe harmonic and then specialized to the first harmonic, which is the main branch-unambiguous readout used throughout the paper.

\subsection{Noise model and harmonic SNR}
\label{subsec:noise_model}

Following Section~\ref{sec:estimators}, consider the calibrated demodulated harmonic phasor $\widehat{\rho}_{21}^{(n)}$. We model it as
\begin{equation}
\widehat{\rho}_{21}^{(n)}=\rho_{21}^{(n)}(\Phi_S)+\eta^{(n)},
\label{eq:additive_noise}
\end{equation}
where $\eta^{(n)}=\eta_I^{(n)}+i\eta_Q^{(n)}$ is circular complex Gaussian noise after demodulation,
\[
\eta_I^{(n)}\sim \mathcal{N}(0,\sigma_n^2),\qquad
\eta_Q^{(n)}\sim \mathcal{N}(0,\sigma_n^2),\qquad
\mathbb{E}[|\eta^{(n)}|^2]=2\sigma_n^2.
\]
The variance $\sigma_n^2$ is the post-demodulation quadrature-noise variance and therefore absorbs detector noise, electronics noise, and the effect of the chosen integration time.

The corresponding harmonic SNR is
\begin{equation}
\SNR_n
=
\frac{|\rho_{21}^{(n)}(\Phi_S)|^2}{\mathbb{E}[|\eta^{(n)}|^2]}
=
\frac{|P_{21}^{(n)}|^2}{2\sigma_n^2},
\label{eq:snr_def}
\end{equation}
where the second equality follows from the exact phase law in Section~\ref{sec:pss}. Hence $\SNR_n$ is independent of the unknown carrier phase $\Phi_S$.

\subsection{Phase and amplitude RMSE laws}
\label{subsec:rmse_laws}

Appendix~\ref{app:rmse_derivation} gives the first-order derivation of the error laws below. The approximation is local and high-SNR: it assumes a nonzero chosen harmonic, small phasor perturbations, a locally smooth response map, and the absence of strong phase wrapping or inverse-map ambiguity over the error scale of interest.

Linearizing the argument of a noisy complex phasor yields
\begin{equation}
\mathrm{Var}(\widehat{\Phi}_S)\approx \frac{1}{2n^2\,\SNR_n},
\qquad
\RMSE_{\Phi}\approx \frac{1}{n\sqrt{2\,\SNR_n}}.
\label{eq:phase_rmse}
\end{equation}
Thus, phase recovery accuracy is controlled solely by the SNR of the chosen phase-bearing harmonic.

For amplitude recovery, linearizing the inverse response map yields
\begin{equation}
\frac{\RMSE_{\Omega_S}}{\Omega_S}
\approx
\frac{1}{|s_n(\Omega_S)|\sqrt{2\,\SNR_n}},
\label{eq:amp_rmse}
\end{equation}
where $s_n(\Omega_S)$ is the logarithmic sensitivity defined in (\ref{eq:log_sensitivity}). In contrast to phase recovery, amplitude recovery depends on both harmonic SNR and the local slope of the response map.

It is useful to package this dependence into the local effective amplitude SNR
\begin{equation}
\SNR_{A,n}^{\mathrm{eff}}(\Omega_S)
\triangleq
|s_n(\Omega_S)|^2\,\SNR_n,
\label{eq:snr_A_eff_uniform}
\end{equation}
so that (\ref{eq:amp_rmse}) becomes
\[
\frac{\RMSE_{\Omega_S}}{\Omega_S}
\approx
\frac{1}{\sqrt{2\,\SNR_{A,n}^{\mathrm{eff}}(\Omega_S)}}.
\]
The familiar proportional-response law
\[
\frac{\RMSE_{\Omega_S}}{\Omega_S}\approx \frac{1}{\sqrt{2\,\SNR_n}}
\]
is therefore recovered only when $|s_n(\Omega_S)|\approx 1$.

\subsection{Mixing-angle design criteria}
\label{subsec:three_optima}

The general RMSE laws above apply to any harmonic $n$. For operating-point design, however, we specialize to the first harmonic $n=1$, which is the main branch-unambiguous signal-bearing harmonic in the present receiver. For a chosen design signal level $\Omega_{S,0}$ and fixed optical settings, define
\begin{equation}
m(\Omega_S,\theta)\triangleq |P_{21}^{(1)}(\Omega_S;\theta)|,
\qquad
s(\Omega_S,\theta)\triangleq
\frac{\Omega_S}{m(\Omega_S,\theta)}
\frac{\partial m}{\partial \Omega_S}(\Omega_S,\theta).
\label{eq:m_and_s_theta}
\end{equation}
It is also convenient to introduce the two local design metrics
\begin{equation}
M_{\phi}(\theta)\triangleq m(\Omega_{S,0},\theta),
\qquad
M_A(\theta)\triangleq
\Omega_{S,0}\left|\frac{\partial m}{\partial \Omega_S}(\Omega_{S,0},\theta)\right|
=
|s(\Omega_{S,0},\theta)|\,m(\Omega_{S,0},\theta).
\label{eq:design_metrics}
\end{equation}
The first metric controls phase accuracy, while the second controls relative amplitude accuracy.

\paragraph{Phase-optimal mixing angle.}
At fixed post-demodulation noise variance, minimizing the phase RMSE in (\ref{eq:phase_rmse}) is equivalent to maximizing the first-harmonic magnitude:
\begin{equation}
\theta_{\phi}^{\star}(\Omega_{S,0})
=
\arg\max_{\theta\in[0,\pi/4]} m(\Omega_{S,0},\theta)
=
\arg\max_{\theta\in[0,\pi/4]} M_{\phi}(\theta).
\label{eq:theta_phase_star}
\end{equation}

\paragraph{Amplitude-optimal mixing angle.}
At fixed noise variance, minimizing the relative amplitude RMSE in (\ref{eq:amp_rmse}) is equivalent to maximizing the local slope magnitude, or equivalently the metric $M_A(\theta)$:
\begin{equation}
\theta_{A}^{\star}(\Omega_{S,0})
=
\arg\max_{\theta\in[0,\pi/4]}
\left[
|s(\Omega_{S,0},\theta)|\,m(\Omega_{S,0},\theta)
\right]
=
\arg\max_{\theta\in[0,\pi/4]}
\left[
\Omega_{S,0}\left|\frac{\partial m}{\partial \Omega_S}(\Omega_{S,0},\theta)\right|
\right].
\label{eq:theta_amp_star}
\end{equation}

\paragraph{Joint-design mixing angle.}
There is no universal single optimum for simultaneous phase and amplitude recovery, because the answer depends on how the two error metrics are weighted. A convenient dimensionless objective is the weighted sum of squared phase RMSE and squared relative amplitude RMSE:
\begin{equation}
J(\theta;\Omega_{S,0})
=
w_{\phi}\,\RMSE_{\Phi}^{2}(\theta;\Omega_{S,0})
+
w_{A}\left(\frac{\RMSE_{\Omega_S}(\theta;\Omega_{S,0})}{\Omega_{S,0}}\right)^{2},
\qquad
w_{\phi},w_A\ge 0.
\label{eq:joint_cost}
\end{equation}
Only the relative weighting of $w_{\phi}$ and $w_A$ matters; if both are nonzero, one may normalize them, for example by imposing $w_{\phi}+w_A=1$, without changing the minimizer. These weights represent application priorities, such as favoring coherent phase recovery, amplitude demodulation fidelity, or a balanced operating regime.

Using (\ref{eq:phase_rmse}) and (\ref{eq:amp_rmse}) with $n=1$ gives
\begin{equation}
J(\theta;\Omega_{S,0})
\approx
\frac{\sigma_1^{2}}{m^{2}(\Omega_{S,0},\theta)}
\left[
w_{\phi}+\frac{w_A}{|s(\Omega_{S,0},\theta)|^{2}}
\right].
\label{eq:joint_cost_explicit}
\end{equation}
The corresponding joint-design mixing angle is
\begin{equation}
\theta_{J}^{\star}(\Omega_{S,0})
=
\arg\min_{\theta\in[0,\pi/4]} J(\theta;\Omega_{S,0}).
\label{eq:theta_joint_star}
\end{equation}

\paragraph{Balanced mixing angle.}
A useful weight-free compromise is the mixing angle that equalizes the \emph{fractional} degradation of phase and amplitude estimation relative to their individually optimal settings. Define the phase-degradation factor
\begin{equation}
D_{\phi}(\theta;\Omega_{S,0})
\triangleq
\frac{\RMSE_{\Phi}(\theta;\Omega_{S,0})}
{\RMSE_{\Phi}(\theta_{\phi}^{\star};\Omega_{S,0})}
=
\frac{M_{\phi}(\theta_{\phi}^{\star})}{M_{\phi}(\theta)},
\label{eq:Dphi_def}
\end{equation}
and the amplitude-degradation factor
\begin{equation}
D_{A}(\theta;\Omega_{S,0})
\triangleq
\frac{\left(\RMSE_{\Omega_S}(\theta;\Omega_{S,0})/\Omega_{S,0}\right)}
{\left(\RMSE_{\Omega_S}(\theta_{A}^{\star};\Omega_{S,0})/\Omega_{S,0}\right)}
=
\frac{M_{A}(\theta_{A}^{\star})}{M_{A}(\theta)},
\label{eq:DA_def}
\end{equation}
where the equalities use the high-SNR laws under fixed post-demodulation noise variance. The balanced angle is then defined by the minimax criterion
\begin{equation}
\theta_{\mathrm{bal}}^{\star}(\Omega_{S,0})
\triangleq
\arg\min_{\theta\in[0,\pi/4]}
\max\!\left\{
D_{\phi}(\theta;\Omega_{S,0}),
D_{A}(\theta;\Omega_{S,0})
\right\}.
\label{eq:theta_bal_star}
\end{equation}
When the two degradation curves intersect in the interior of the interval, the minimax solution occurs at or very near the crossing
\[
D_{\phi}(\theta_{\mathrm{bal}}^{\star};\Omega_{S,0})
=
D_{A}(\theta_{\mathrm{bal}}^{\star};\Omega_{S,0}).
\]
Unlike $\theta_{J}^{\star}$, which depends on user-specified weights, $\theta_{\mathrm{bal}}^{\star}$ provides a practical weight-free compromise that equalizes the relative phase and amplitude penalties. In the numerical section, we use $\theta_0=\theta_{\mathrm{bal}}^{\star}(\Omega_{S,0})$ as the default nominal operating point and reserve $\theta=\pi/8$ for perturbative-seed comparisons only.

\subsection{Perturbative initialization rule}
\label{subsec:perturbative_seed}

In the weak-signal, weak-loop-closure regime, Appendix~\ref{app:perturbative_scaling} gives the representative scaling
\begin{equation}
m(\Omega_S,\theta)
\propto
\Omega_p\,\Omega_c^2\,\Omega_S\,
f(\theta)\,
\mathcal{F}(\Delta_p,\Delta_c,\{\gamma\}),
\qquad
f(\theta)=\sin\theta\,\cos\theta\,\cos(2\theta).
\label{eq:perturbative_scaling}
\end{equation}
Since
\begin{equation}
f(\theta)=\frac{1}{4}\sin(4\theta)=\frac{\beta}{2(1+\beta^2)},
\label{eq:f_theta}
\end{equation}
the perturbative initialization rule is
\begin{equation}
\theta_{\mathrm{seed}}=\frac{\pi}{8},
\qquad
\beta_{\mathrm{seed}}=1,
\qquad
E_{z,\mathrm{seed}}=\frac{\hbar\Delta_{34}}{2|\mu_{34}^{z}|}.
\label{eq:theta_seed}
\end{equation}
This result is only a weak-coupling heuristic seed for numerical optimization of the full model. It is not a design theorem and should not be interpreted as a universal optimum.

In general, non-perturbative corrections shift the optimum of the full Floquet--Liouville model, and amplitude sensitivity can depart substantially from the linear-response regime. For that reason, the practically relevant operating points are $\theta_{\phi}^{\star}$, $\theta_{A}^{\star}$, and, once application weights are specified, $\theta_{J}^{\star}$.

All of the criteria above are defined with respect to the spatially uniform-bias benchmark. The next section explains how spatial bias variation modifies the coherent gain and the effective amplitude sensitivity through quasistatic averaging.

\section{Uniform benchmark and nonuniform-bias effects}
\label{sec:bias_cases}

The error laws and design criteria of Section~\ref{sec:performance} are benchmark results for a spatially uniform DC bias. We now consider how these results are modified when the bias varies across the interaction region. The key idea is to treat the nonuniform case as a quasistatic weighted average of local uniform-bias responses.

\subsection{Uniform-bias benchmark}
\label{subsec:uniform_benchmark}

The central benchmark of the paper is the uniform-bias model
\begin{equation}
E_z(\mathbf{r})\equiv E_z.
\label{eq:uniform_benchmark_repeat}
\end{equation}
In that setting, the reduced Hamiltonian has the exact Floquet form (\ref{eq:Heff_main}), the periodic steady state obeys the exact harmonic phase law (\ref{eq:factorization}), and the estimators and high-SNR laws of Sections~\ref{sec:estimators} and \ref{sec:performance} apply directly. When we later compare against a nonuniform bias centered at a nominal value $E_0$, the corresponding uniform-bias reference is the response evaluated at $E_z=E_0$.

\subsection{Quasistatic nonuniform bias and coherent averaging}
\label{subsec:nonuniform_bias}

We now allow
\begin{equation}
E_z(\mathbf{r})=E_0+\delta E(\mathbf{r}),
\qquad
\langle \delta E\rangle=0,
\qquad
\langle \delta E^2\rangle=\sigma_E^2,
\label{eq:inhomogeneous_bias}
\end{equation}
where the angle brackets denote a normalized spatial average over the interaction region,
\begin{equation}
\langle f\rangle
\triangleq
\int_V w(\mathbf{r})\,f(\mathbf{r})\,d^3\mathbf{r},
\qquad
\int_V w(\mathbf{r})\,d^3\mathbf{r}=1.
\label{eq:spatial_average_def}
\end{equation}
The weighting function $w(\mathbf{r})$ represents the spatial sensitivity of the optical measurement, including mode overlap and collection weighting.

The quasistatic averaging model assumes that: (i) the bias profile is effectively static over the demodulation window, (ii) each spatial point can be described by the corresponding local uniform-bias response, and (iii) the measured probe harmonic is a linear weighted average of those local coherences. Under these assumptions, the local first-harmonic coefficient obeys
\[
\rho_{21}^{(1)}(\Phi_S;E_z)=P_{21}^{(1)}(E_z)\,e^{i\Phi_S}.
\]
Averaging over the interaction region gives
\begin{equation}
\overline{\rho}_{21}^{(1)}(\Phi_S)
=
\left\langle P_{21}^{(1)}(E_z(\mathbf{r}))\right\rangle e^{i\Phi_S}
\triangleq
\overline{P}_{21}^{(1)}\,e^{i\Phi_S}.
\label{eq:averaged_harmonic}
\end{equation}
This is a key structural result: because every local first-harmonic response carries the same factor $e^{i\Phi_S}$, any linear spatial average preserves the exact phase law. Nonuniform bias alters the complex coefficient, but not the multiplicative phase factor.

To compare the averaged response with the corresponding uniform benchmark at the nominal bias $E_0$, define the coherent-gain factor
\begin{equation}
G\triangleq \frac{|\overline{P}_{21}^{(1)}|}{|P_{21}^{(1)}(E_0)|}.
\label{eq:G_def}
\end{equation}
At fixed detection-noise variance,
\begin{equation}
\SNR_{1,\mathrm{eff}}=G^2\,\SNR_{1,0},
\qquad
\SNR_{1,0}=\frac{|P_{21}^{(1)}(E_0)|^2}{2\sigma_1^2}.
\label{eq:snr_eff}
\end{equation}
Hence the high-SNR phase penalty relative to the uniform benchmark is
\begin{equation}
\frac{\RMSE_{\Phi,\mathrm{nonunif}}}{\RMSE_{\Phi,\mathrm{unif}}}
\approx \frac{1}{G}.
\label{eq:phase_penalty}
\end{equation}
For phase estimation, the nonuniform-bias effect therefore enters through coherent-gain reduction alone.

Two physical mechanisms dominate the change in $G$. First, the local mixing angle $\theta(E_z)$ varies across the interaction region, which changes both the bias-enabled optical leg $\OmTwoFour$ and the effective RF coupling $\OmThreeFour$ from point to point. Second, the dressed splitting $\omega_{34}^{(S)}(E_z)$ varies with the local bias, which creates inhomogeneous RF detuning even if the received signal is tuned to the nominal dressed splitting at $E_0$:
\begin{equation}
\delta\Delta_S(\mathbf{r})
\approx
-\left.\frac{d\omega_{34}^{(S)}}{dE_z}\right|_{E_0}\delta E(\mathbf{r}),
\qquad
\left.\frac{d\omega_{34}^{(S)}}{dE_z}\right|_{E_0}
=
\frac{2|\mu_{34}^{z}|}{\hbar}\sin(2\theta_0),
\quad
\theta_0=\theta(E_0).
\label{eq:domega_dE}
\end{equation}

\subsection{Amplitude consequences under nonuniform bias}
\label{subsec:amplitude_nonuniform}

Amplitude recovery is affected more subtly than phase recovery because spatial averaging modifies not only the coherent harmonic gain but also the slope of the effective response map. Define the averaged first-harmonic response map
\begin{equation}
m_{\mathrm{avg}}(\Omega_S)\triangleq \left|\overline{P}_{21}^{(1)}(\Omega_S)\right|,
\label{eq:effective_response_map}
\end{equation}
together with its logarithmic sensitivity
\begin{equation}
s_{\mathrm{avg}}(\Omega_S)\triangleq
\frac{d\ln m_{\mathrm{avg}}(\Omega_S)}{d\ln \Omega_S}.
\label{eq:effective_sensitivity}
\end{equation}
Compared with the uniform case, spatial averaging can shift, shrink, or even destroy the monotone interval on which the response map is injective, so the usable amplitude-estimation range itself can change.

The high-SNR amplitude law becomes
\begin{equation}
\frac{\RMSE_{\Omega_S}}{\Omega_S}
\approx
\frac{1}{|s_{\mathrm{avg}}(\Omega_S)|\sqrt{2\,\SNR_{1,\mathrm{eff}}}}.
\label{eq:amp_rmse_nonuniform}
\end{equation}
For a chosen design signal level $\Omega_{S,0}$, it is convenient to define the corresponding effective amplitude SNR
\begin{equation}
\SNR_{A,\mathrm{eff}}
\triangleq
|s_{\mathrm{avg}}(\Omega_{S,0})|^2\,\SNR_{1,\mathrm{eff}},
\label{eq:snr_A_eff_theory}
\end{equation}
so that the local high-SNR form at $\Omega_{S,0}$ becomes
\[
\frac{\RMSE_{\Omega_S}}{\Omega_{S,0}}
\approx
\frac{1}{\sqrt{2\,\SNR_{A,\mathrm{eff}}}}.
\]

Unlike phase recovery, amplitude recovery is therefore controlled by two factors: coherent-gain reduction through $G$ and slope reshaping through $s_{\mathrm{avg}}$. In particular, moderate nonuniformity can in some regimes reduce the coherent gain while simultaneously moving the operating point onto a steeper effective response branch. As a result, the amplitude behavior at fixed baseline SNR is not determined by $G$ alone.

The numerical section quantifies these effects directly: it confirms the uniform-bias design optima, measures the coherent-gain reduction relevant for phase recovery, and tests the slope-adjusted collapse laws for amplitude recovery under nonuniform bias.

\section{Numerical illustrations of the phase law, design trade-offs, and nonuniform-bias effects}
\label{sec:numerics}

This section uses the reduced Floquet--Liouville model to illustrate the main claims of the theory. The numerical results are organized around four questions: whether the first harmonic carries the received-signal phase exactly, whether its magnitude provides a usable amplitude-response map, whether the full model supports distinct uniform-bias design optima, and how nonuniform bias modifies phase and amplitude recovery. To keep the discussion focused on these points, internal solver-convergence and time-domain validation checks are omitted from the main text and reported in Appendix~\ref{app:numerics_validation}.

\subsection{Numerical setup}
\label{subsec:numerical_setup}

We compute the periodic steady state by truncating the harmonic-balance equations~\eqref{eq:harmonic_balance_eq} to harmonics $n\in[-N,N]$ and solving the resulting block-tridiagonal linear system for the reference-phase coefficients $P^{(n)}$ at $\Phi_S=0$. Unless stated otherwise, all production results use $N=3$. Higher-order truncations and direct time-domain integration were checked separately and found to agree with the production setting over the parameter range considered here; see Appendix~\ref{app:numerics_validation}.

Dissipation is modeled by the minimal ladder-EIT Lindblad set used throughout the numerical code: decay $\ket{2}\to\ket{1}$ at rate $\gamma_{21}$, decay $\ket{3^{(S)}}\to\ket{2}$ and $\ket{4^{(S)}}\to\ket{2}$ at rates $\gamma_{32}$ and $\gamma_{42}$, and weak pure dephasing on the upper states. All frequencies are reported in units of $\gamma_{21}$, which is set to unity.

Unless the horizontal axis itself is the mixing angle or the nonuniformity level, we use the nominal operating point
\begin{equation}
\Omega_p=0.2,\qquad
\Omega_c=1,\qquad
\Delta_p=\Delta_c=0,\qquad
\Omega_{S,0}=0.12,\qquad
\Delta_S=0,\qquad
\theta_0=\theta_{\mathrm{bal}}^{\star}(\Omega_{S,0}).
\label{eq:nominal_point}
\end{equation}
For the present parameter set, the balanced criterion places the nominal operating point near $\theta_0\approx 0.56$ rad. This balanced angle is used in the remaining numerical examples, whereas the perturbative seed $\theta=\pi/8$ is retained only in the design-landscape comparison. Monte-Carlo RMSE curves below use $3\times 10^4$ trials per SNR point.

\subsection{Exact phase law and first-harmonic amplitude response}
\label{subsec:numerical_phase_amp}

Figure~\ref{fig:num_phase_law} gives the most direct numerical illustration of the reception mechanism at the balanced nominal point $\theta_0=\theta_{\mathrm{bal}}^{\star}(\Omega_{S,0})$. At this operating point, the bias-enabled upper optical leg is active and the first-harmonic phasor has a substantial nonzero magnitude. By contrast, when $\theta=0$ the loop-closing leg vanishes and the first harmonic collapses to zero in the reduced model. Varying $\Phi_S$ then rotates $\rho_{21}^{(1)}(\Phi_S)$ rigidly in the complex plane without changing its radius, exactly as predicted by the phase law~\eqref{eq:phase_law_probe}. The residual plot in Fig.~\ref{fig:num_phase_law}(b) makes the point directly: after division by a reference harmonic, the recovered first-harmonic phase tracks $\Phi_S$ with residuals at the level of numerical machine precision across the full $2\pi$ sweep.

\begin{figure}[t]
\centering
\begin{minipage}{0.49\linewidth}
\centering
\includegraphics[width=\linewidth]{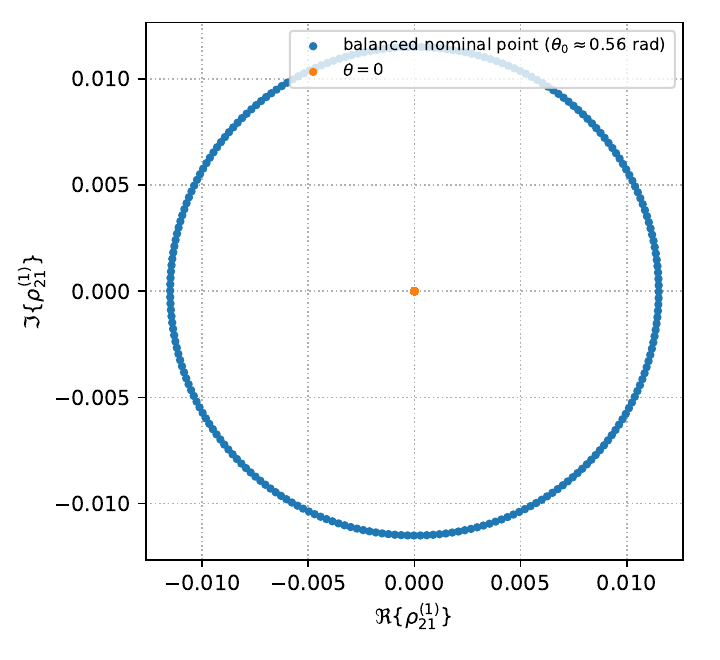}\\[-0.5em]
\textbf{(a)} Complex phasor $\rho_{21}^{(1)}(\Phi_S)$ for $\Phi_S\in[0,2\pi)$ at $\theta=\theta_0$, with the $\theta=0$ collapse shown for comparison.
\end{minipage}
\hfill
\begin{minipage}{0.49\linewidth}
\centering
\includegraphics[width=\linewidth]{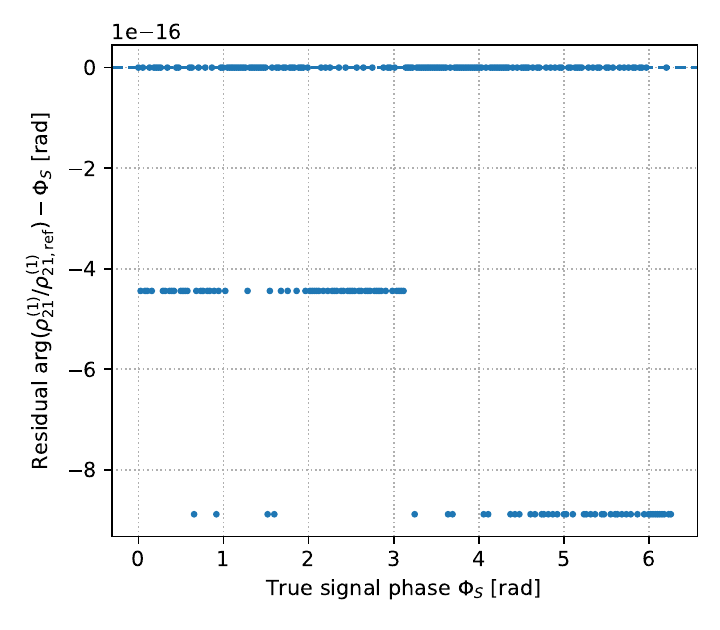}\\[-0.5em]
\textbf{(b)} Phase-law residual $\arg\!\big(\rho_{21}^{(1)}/\rho_{21,\mathrm{ref}}^{(1)}\big)-\Phi_S$.
\end{minipage}
\caption{Numerical confirmation of the exact first-harmonic phase law at the balanced nominal operating point. The DC-Stark-enabled loop generates a nonzero first harmonic, and varying $\Phi_S$ rotates that harmonic rigidly in the complex plane. The residual plot shows that the reference-normalized first-harmonic phase follows the received-signal phase with numerically negligible deviation.}
\label{fig:num_phase_law}
\end{figure}

Amplitude recovery is illustrated in Fig.~\ref{fig:num_amp_calibration}. Panel~(a) shows the first-harmonic response map
\[
m(\Omega_S)\triangleq |P_{21}^{(1)}(\Omega_S)|,
\]
while panel~(b) shows the corresponding logarithmic sensitivity
\[
s(\Omega_S)\triangleq \frac{d\ln m(\Omega_S)}{d\ln \Omega_S}.
\]
At the balanced operating point, the design level $\Omega_{S,0}$ lies on an injective monotone branch, and the local sensitivity is moderate, with $|s(\Omega_{S,0})|\approx 0.65$. This is substantially steeper than at the perturbative seed $\theta=\pi/8$, so the balanced operating point yields markedly improved amplitude conditioning while preserving a clean invertible response branch.

\begin{figure}[t]
\centering
\includegraphics[width=0.92\linewidth]{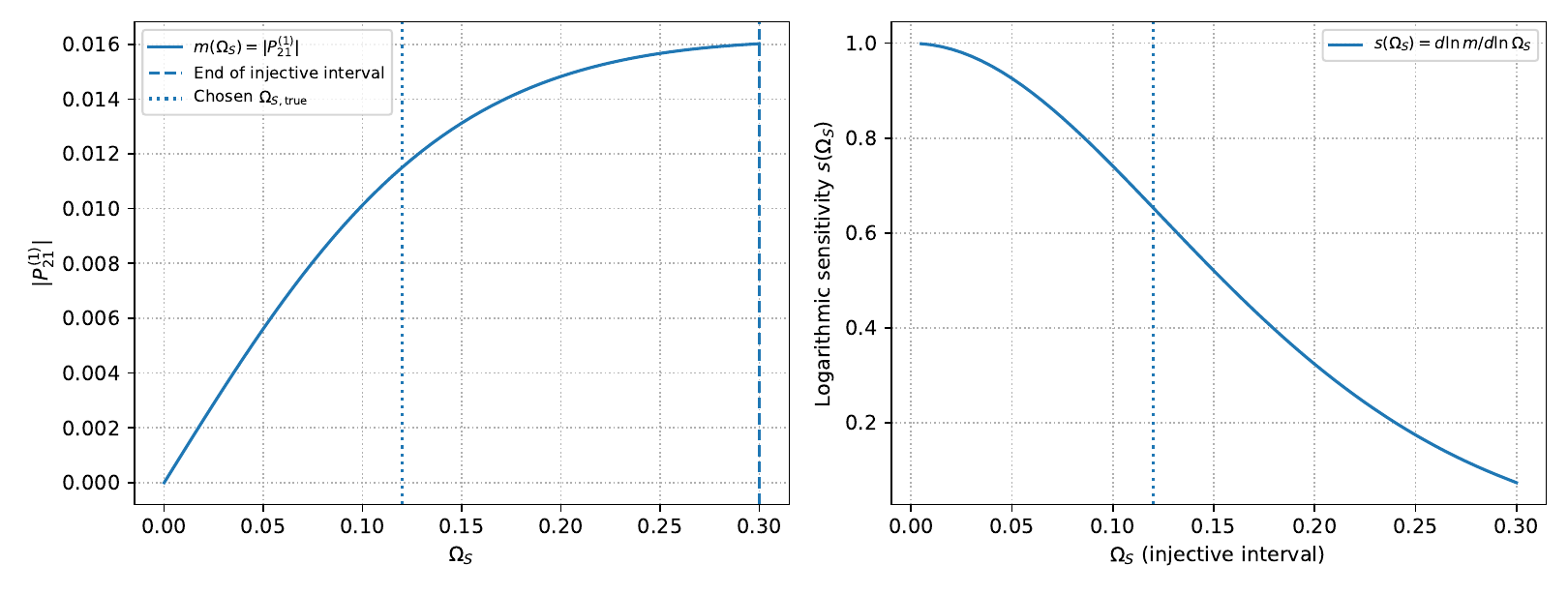}
\caption{First-harmonic response map and local logarithmic sensitivity at the nominal operating point. Panel (a) shows the injective branch of $m(\Omega_S)=|P_{21}^{(1)}(\Omega_S)|$ together with the chosen design level $\Omega_{S,0}=0.12$. Panel (b) shows $s(\Omega_S)=d\ln m/d\ln\Omega_S$ on the same branch. Amplitude recovery requires both injectivity of the response map and a sufficiently large local slope.}
\label{fig:num_amp_calibration}
\end{figure}

\subsection{Uniform-bias design landscape}
\label{subsec:numerical_theta_sweep}

The next figure shows why the perturbative seed $\theta=\pi/8$ should not be confused with an operating optimum of the full model. Figure~\ref{fig:num_design_landscape}(a) compares the full first-harmonic phase metric
\[
M_{\phi}(\theta)=m(\Omega_{S,0},\theta)
\]
with the scaled perturbative proxy $f(\theta)=\sin\theta\cos\theta\cos(2\theta)$. The full-model phase optimum occurs at
\[
\theta_{\phi}^{\star}\approx 0.49\ \text{rad},
\]
which is clearly to the right of the perturbative seed $\pi/8\approx 0.393$ rad. The proxy therefore remains useful as an initialization rule, but it does not locate the full-model optimum.

Figure~\ref{fig:num_design_landscape}(b) overlays the normalized phase metric $M_{\phi}(\theta)$, the amplitude metric
\[
M_A(\theta)=\Omega_{S,0}\left|\frac{\partial m}{\partial\Omega_S}(\Omega_{S,0},\theta)\right|,
\]
and a representative equal-weight joint utility derived from $J(\theta)$. The amplitude-optimal angle occurs later than the phase-optimal one, at approximately
\[
\theta_{A}^{\star}\approx 0.60\ \text{rad}.
\]
The balanced operating point $\theta_{\mathrm{bal}}^{\star}$ is read off near the crossing of the normalized phase and amplitude curves; for the present parameter set it lies at approximately
\[
\theta_{\mathrm{bal}}^{\star}\approx 0.56\ \text{rad},
\]
between $\theta_{\phi}^{\star}$ and $\theta_{A}^{\star}$. The equal-weight joint-design angle $\theta_{J}^{\star}$ lies nearby. In the remainder of the numerical section, we therefore use $\theta_0=\theta_{\mathrm{bal}}^{\star}(\Omega_{S,0})$ as the nominal operating point rather than the perturbative seed.

\begin{figure}[t]
\centering
\includegraphics[width=0.92\linewidth]{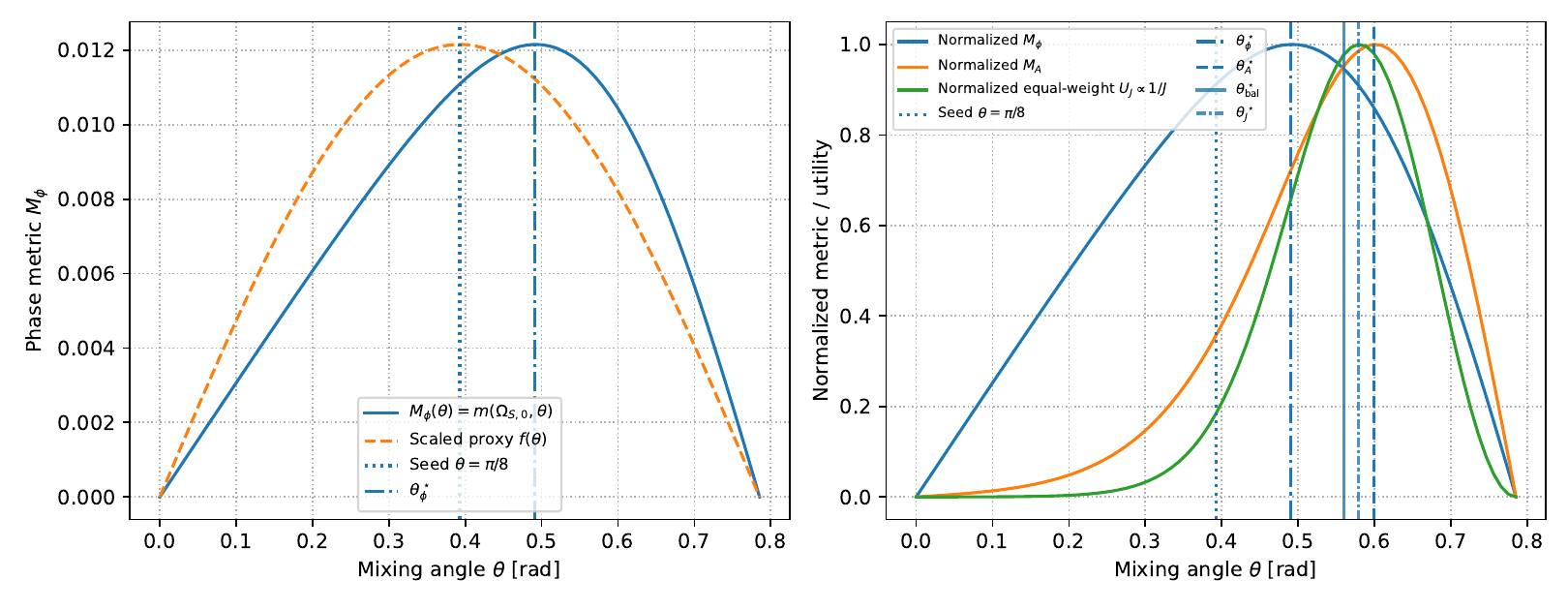}
\caption{Uniform-bias design landscape at the nominal design level $\Omega_{S,0}$. Panel (a) compares the full first-harmonic phase metric $M_{\phi}(\theta)=m(\Omega_{S,0},\theta)$ with the scaled perturbative proxy $f(\theta)$, showing that the full-model phase optimum is shifted away from the perturbative seed. Panel (b) overlays the normalized phase metric $M_{\phi}$, amplitude metric $M_A$, and a representative equal-weight joint utility derived from $J(\theta)$. The crossing of the normalized phase and amplitude curves indicates the balanced operating point $\theta_{\mathrm{bal}}^{\star}$, which lies between the separate phase- and amplitude-optimal angles.}
\label{fig:num_design_landscape}
\end{figure}

\subsection{Uniform-bias RMSE versus harmonic SNR}
\label{subsec:numerical_rmse_uniform}

To validate the high-SNR laws of Section~\ref{sec:performance}, we adopt the harmonic-domain observation model
\begin{equation}
\widehat{\rho}_{21}^{(1)}
=
P_{21}^{(1)}e^{i\Phi_S}+\eta,
\qquad
\eta\sim\mathcal{CN}(0,2\sigma^2),
\qquad
\SNR_1=\frac{|P_{21}^{(1)}|^2}{2\sigma^2}.
\label{eq:mc_model}
\end{equation}
For each SNR value, we draw $\Phi_S\sim\mathrm{Unif}[0,2\pi)$, add independent complex Gaussian noise, apply the phase estimator~\eqref{eq:phase_estimator}, and recover the amplitude by inverting the first-harmonic response map on its injective branch.

The results are shown in Fig.~\ref{fig:num_rmse_homog}. The phase RMSE follows the high-SNR law~\eqref{eq:phase_rmse} closely across the plotted range. The amplitude behavior remains distinct, but it is now substantially improved relative to the seed-angle benchmark because the balanced operating point removes much of the original slope penalty associated with $\theta=\pi/8$. At the present nominal point, the remaining offset between the phase and amplitude curves is explained by the fact that the local sensitivity is still below unity, with $|s(\Omega_{S,0})|\approx 0.65$, so the amplitude curve continues to follow the sensitivity-corrected law~\eqref{eq:amp_rmse} rather than the proportional-response limit.

\begin{figure}[t]
\centering
\includegraphics[width=0.90\linewidth]{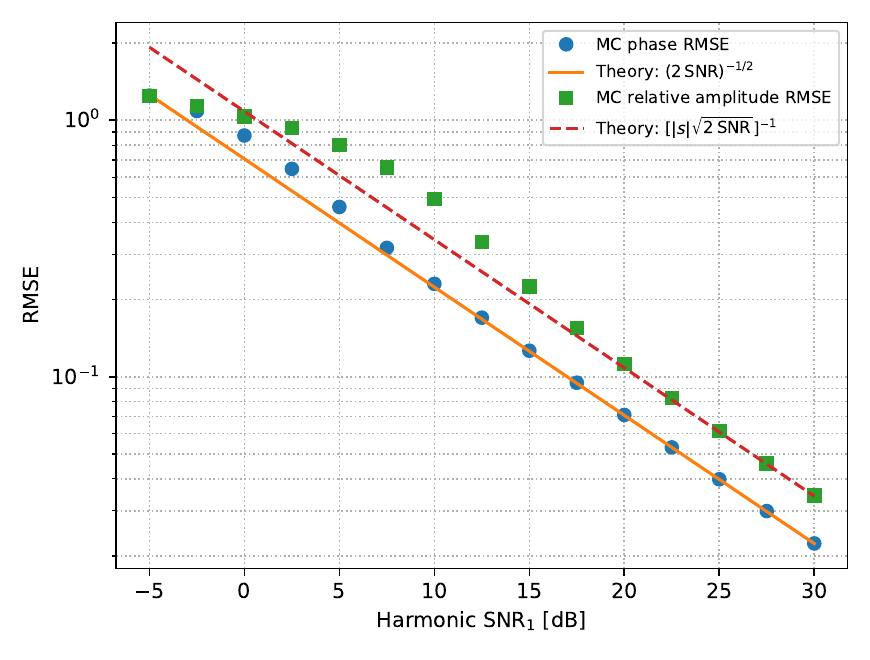}
\caption{Uniform-bias estimator performance versus harmonic SNR. Markers show Monte-Carlo RMSE for the phase and relative amplitude estimates under the observation model~\eqref{eq:mc_model}; curves show the high-SNR predictions~\eqref{eq:phase_rmse} and~\eqref{eq:amp_rmse}. The remaining amplitude penalty at fixed harmonic SNR is the direct consequence of a local response slope below unity at the balanced nominal operating point.}
\label{fig:num_rmse_homog}
\end{figure}

\subsection{Nonuniform-bias effects}
\label{subsec:numerical_nonuniform}

We now turn to the quasistatic nonuniform-bias model of Section~\ref{sec:bias_cases}. Following the numerical generator, spatial nonuniformity is parameterized by $\beta=\tan(2\theta)$, and the relative spread $\sigma_{\beta}/\beta_0$ is used as the dimensionless nonuniformity level. The numerical results below use three representative cases,
\[
\sigma_{\beta}/\beta_0=0.01,\ 0.02,\ 0.05,
\]
and compare them against the uniform-bias reference evaluated at the balanced nominal point $\theta_0=\theta_{\mathrm{bal}}^{\star}(\Omega_{S,0})$.

For the three representative nonuniformity levels, the corresponding coherent-gain factors are
\[
G\approx 0.32,\ 0.22,\ 0.10
\qquad\text{for}\qquad
\sigma_{\beta}/\beta_0=0.01,\ 0.02,\ 0.05,
\]
respectively. Figure~\ref{fig:num_inhomog_collapse} shows how this coherent averaging enters the phase and amplitude estimators. Panels~(a) and~(b) confirm the phase prediction of Section~\ref{sec:bias_cases}: when phase RMSE is plotted against the baseline uniform-reference SNR, the curves separate strongly, but they collapse when the horizontal axis is changed to the coherent-gain-corrected quantity
\[
\SNR_{1,\mathrm{eff}}=G^2\SNR_{1,0}.
\]
This is the numerical signature of the fact that spatial averaging preserves the exact phase law while reducing only the coherent harmonic gain.

Amplitude recovery is more subtle. In the present parameter set, the three nonuniform cases all yield a larger effective logarithmic sensitivity at the design point,
\[
|s_{\mathrm{avg}}(\Omega_{S,0})|\approx 0.93,
\]
compared with the uniform balanced-reference value
\[
|s(\Omega_{S,0})|\approx 0.65.
\]
Thus, for these examples, nonuniform bias reduces coherent gain but simultaneously pushes the effective response map onto a steeper local branch. Panels~(c) and~(d) resolve this competition: the raw amplitude curves do not collapse on the baseline SNR axis, but they approach a common high-SNR behavior when replotted against the sensitivity-adjusted axis
\[
\SNR_{A,\mathrm{eff}}=|s_{\mathrm{avg}}(\Omega_{S,0})|^2\,\SNR_{1,\mathrm{eff}},
\]
exactly as predicted by~\eqref{eq:snr_A_eff_theory}. The remaining spread at low SNR is the expected breakdown of the strictly local linearization.

\begin{figure}[t]
\centering
\begin{minipage}{0.49\linewidth}
\centering
\includegraphics[width=\linewidth]{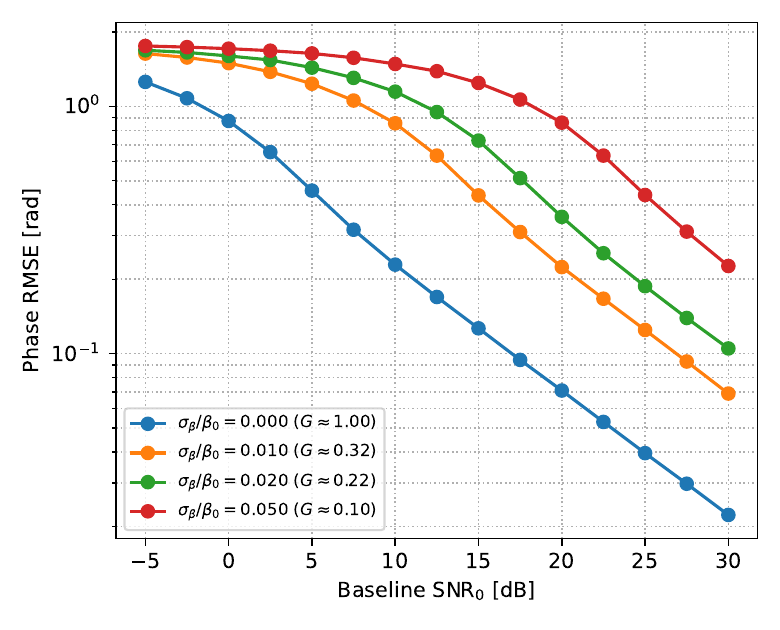}\\[-0.5em]
\textbf{(a)} Phase RMSE versus baseline $\SNR_{1,0}$.
\end{minipage}
\hfill
\begin{minipage}{0.49\linewidth}
\centering
\includegraphics[width=\linewidth]{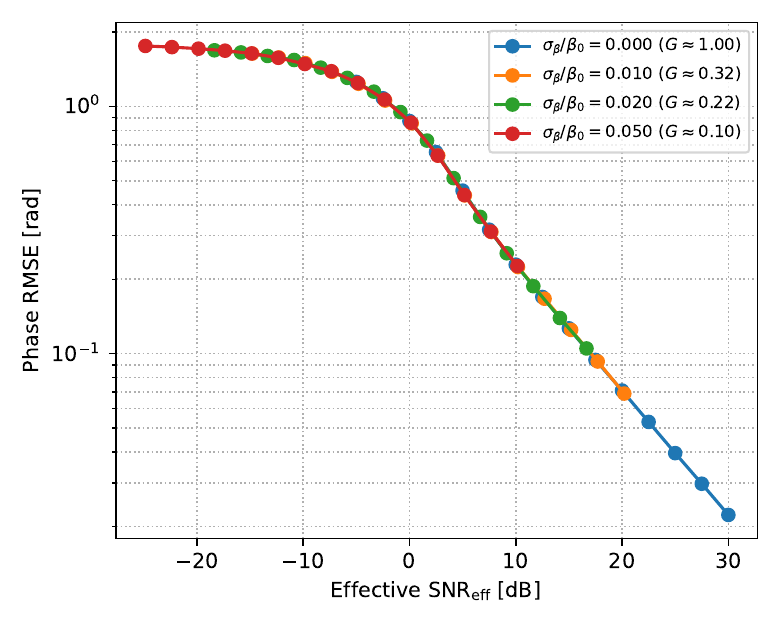}\\[-0.5em]
\textbf{(b)} Phase RMSE versus $\SNR_{1,\mathrm{eff}}$.
\end{minipage}

\vspace{0.6em}

\begin{minipage}{0.49\linewidth}
\centering
\includegraphics[width=\linewidth]{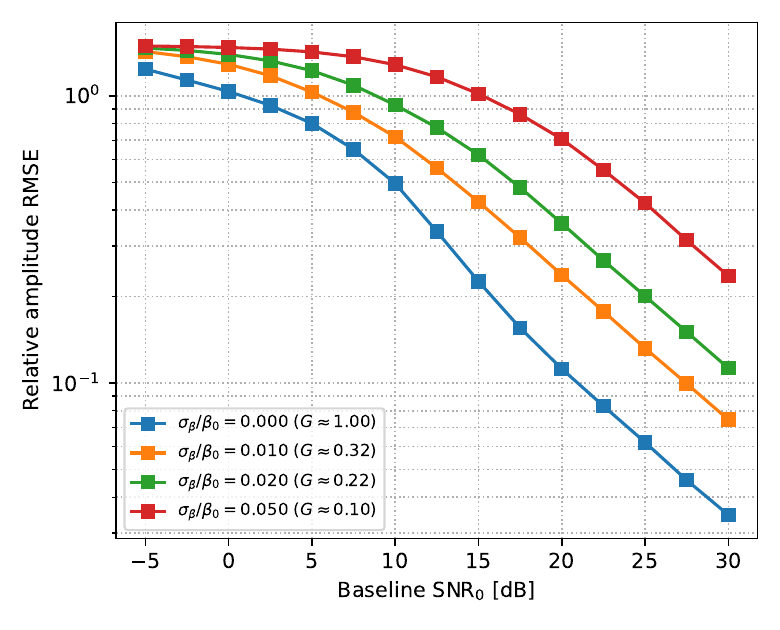}\\[-0.5em]
\textbf{(c)} Relative amplitude RMSE versus baseline $\SNR_{1,0}$.
\end{minipage}
\hfill
\begin{minipage}{0.49\linewidth}
\centering
\includegraphics[width=\linewidth]{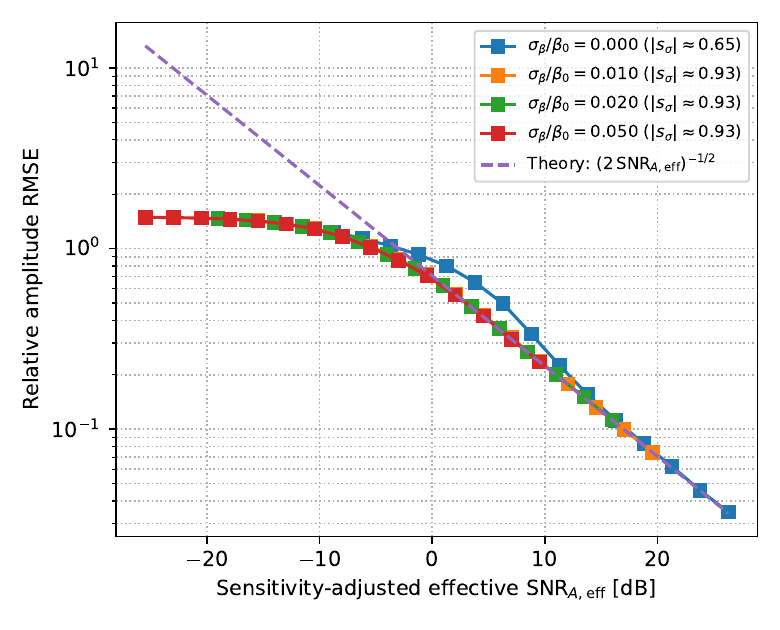}\\[-0.5em]
\textbf{(d)} Relative amplitude RMSE versus $\SNR_{A,\mathrm{eff}}$.
\end{minipage}

\caption{Collapse laws under nonuniform bias around the balanced nominal operating point. When plotted against the baseline uniform-reference SNR, different nonuniformity levels separate the RMSE curves strongly. For phase, the curves collapse when replotted against the coherent-gain-corrected axis $\SNR_{1,\mathrm{eff}}=G^2\SNR_{1,0}$. For amplitude, a comparable high-SNR collapse appears only after the additional slope correction $\SNR_{A,\mathrm{eff}}=|s_{\mathrm{avg}}|^2\SNR_{1,\mathrm{eff}}$ is included.}
\label{fig:num_inhomog_collapse}
\end{figure}

Taken together, these numerical results support the theoretical interpretation of the paper. The first harmonic carries the received-signal phase exactly, its magnitude provides an invertible amplitude-response map on a usable operating branch, the full model separates phase-, amplitude-, balanced-, and joint-design operating points under uniform bias, and nonuniform bias acts through coherent-gain reduction for phase but through both gain and slope for amplitude.

\section{Discussion}
\label{sec:discussion}

The central conceptual result of this paper is a change in the resource trade-off for phase-sensitive Rydberg reception. In existing approaches, sensitivity to the carrier phase of a received RF field is typically obtained by injecting an RF local oscillator into the atoms or by adding extra time-varying optical/RF reference tones. Here, the missing loop-closure resource is instead supplied by a \emph{static} DC bias. By Stark-mixing a near-degenerate upper pair, the bias activates a second optical upper leg and thereby enables phase-sensitive reception of a single received RF field without an auxiliary RF LO in the atoms. The price of that simplification is shifted elsewhere: performance becomes strongly tied to DC-bias control, spatial homogeneity, and calibration of the harmonic response map.

It is also important to separate what is structurally exact from what is model dependent. The exact harmonic phase law in Section~\ref{sec:pss} is robust at the level of system structure: as long as the unknown signal phase enters only through the combination $(\omega_S t+\Phi_S)$ and the dissipator is time independent, each harmonic acquires the factor $e^{in\Phi_S}$. By contrast, the harmonic magnitudes, amplitude-response maps, high-SNR prefactors, and optimal mixing angles are not universal. They depend on the reduced couplings, detunings, linewidths, and operating point, and they will in general shift when the minimal four-level model is replaced by a more complete manifold-level description. In the same spirit, the present framework addresses amplitude and carrier-phase recovery at a known or separately tracked carrier frequency $\omega_S$; it is not a theory of blind carrier acquisition.

From a practical viewpoint, the main design lesson is that the perturbative rule $\theta=\pi/8$ should be viewed only as an initialization heuristic. Actual operation should be selected from the full-model criteria $\theta_{\phi}^{\star}$, $\theta_{A}^{\star}$, or, when a compromise is required, either the weighted joint-design criterion $\theta_{J}^{\star}$ or the balanced criterion $\theta_{\mathrm{bal}}^{\star}$ that equalizes the fractional phase and amplitude penalties. In the numerical examples, we adopt $\theta_{\mathrm{bal}}^{\star}$ as the default nominal operating point for precisely that reason. The nonuniform-bias analysis further shows that phase and amplitude recovery are affected in different ways. Phase estimation degrades primarily through coherent averaging loss, summarized by the coherent-gain factor $G$, whereas amplitude estimation depends on both coherent gain and the reshaping of the effective response-map slope. This means that realistic calibration under experimentally relevant bias maps is not merely a secondary correction; it is part of the receiver design problem itself.

The model remains deliberately minimal. It relies on an isolated-pair Stark approximation and neglects additional nearby Stark states, bias-induced mixing with lower-lying ladder states (whose much larger energy separations make their Stark admixture subleading in the minimal model), and magnetic-sublevel structure (degenerate here in the absence of an applied magnetic field). It also omits Doppler and transit effects, treats bias nonuniformity as quasistatic rather than time fluctuating, and does not yet identify a species-specific implementation. These are natural next steps rather than conceptual obstacles. One direction is manifold-level atomic design: selecting concrete states, polarizations, and bias conditions that realize the reduced couplings used here. The other is device-level modeling and measurement in realistic vapor-cell geometries, including electrode design and field mapping, to determine how closely a practical receiver can approach the uniform-bias benchmark established in this work.

\section{Conclusion}
\label{sec:conclusion}

This paper establishes a theoretical route to recovering the amplitude and carrier phase of a single received RF field with a Rydberg-atom receiver without injecting an RF local oscillator into the atoms. The enabling mechanism is a static DC bias that Stark-mixes a near-degenerate upper pair, activates a second optical upper leg, and closes a phase-sensitive loop within a receiver driven only by the probe field, the coupling field, and the received RF signal itself. Under the uniform-bias benchmark, this leads to a periodically driven reduced model whose steady-state harmonics obey an exact phase law, so that the carrier phase of the received field is transferred directly to probe harmonics.

That structure yields direct phase and amplitude estimators from a demodulated probe harmonic, explicit high-SNR error laws, and distinct phase-optimal, amplitude-optimal, weighted joint-design, and balanced-compromise mixing angles. The analysis also shows how spatial bias nonuniformity modifies practical performance: it preserves the multiplicative phase law, but reduces coherent harmonic gain and reshapes the effective amplitude-response slope. Taken together, these results provide a minimal analytical foundation for implementation-specific design and experimental realization of LO-free phase-sensitive Rydberg receivers at a known or separately tracked carrier frequency.

\appendices

\section{Derivation of the Stark-mixed rotating-frame Hamiltonian}
\label{app:stark_derivation}

This appendix derives the reduced Hamiltonian used in the main text. The construction assumes: (i) the isolated-pair Stark approximation for the near-degenerate upper bare pair $\{\ket{3},\ket{4}\}$, (ii) retention of only the dominant probe, coupling, and signal dipole matrix elements, (iii) the rotating-wave approximation (RWA), and (iv) a final gauge choice that makes the unknown signal phase appear only through the combination $(\omega_S t+\Phi_S)$. Throughout, it is useful to distinguish three related quantities: the physical received-field amplitude $A_S$, the underlying signal Rabi scale $\Omega_S$ defined in~\eqref{eq:rabi_scales}, and the effective Stark-basis coupling $\OmThreeFour$ that appears in the reduced Hamiltonian.

\subsection{Bare Hamiltonian and isolated-pair Stark mixing}

Let $\{\ket{n}\}_{n=1}^{4}$ denote field-free eigenstates with $\braket{n|m}=\delta_{nm}$. The bare Hamiltonian is
\begin{equation}
\Hatom=\hbar\sum_{n=1}^{4}\omega_n\ket{n}\!\bra{n}.
\label{eq:app_Hatom}
\end{equation}
A static bias field $E_z$ along $z$ adds the Stark interaction
\begin{equation}
H_0'=\Hatom-\hat{\mu}_z E_z,\qquad \hat{\mu}_z=\hat{\boldsymbol{\mu}}\!\cdot\!\mathbf{e}_z.
\label{eq:app_stark_hamiltonian}
\end{equation}
For field-free eigenstates of definite parity, the diagonal Stark matrix elements vanish, $\bra{n}\hat{\mu}_z\ket{n}=0$ \cite{Gallagher1994Rydberg}.

In the minimal model, only the upper pair $\{\ket{3},\ket{4}\}$ is mixed, so the relevant block is
\begin{equation}
H'_{34}=
\begin{bmatrix}
\hbar\omega_3 & -E_z\mu_{34}^{z} \\
-E_z\mu_{43}^{z} & \hbar\omega_4
\end{bmatrix},
\qquad
\mu_{43}^{z}=\mu_{34}^{z*}.
\label{eq:app_H34_block}
\end{equation}
Choosing phases so that $\mu_{34}^{z}$ is real and positive, diagonalization of the upper $2\times2$ Stark block $H'_{34}$ yields the Stark states
\begin{align}
\ket{3^{(S)}} &= \cos\theta\,\ket{3}+\sin\theta\,\ket{4}, \nonumber\\
\ket{4^{(S)}} &= -\sin\theta\,\ket{3}+\cos\theta\,\ket{4},
\label{eq:app_stark_states}
\end{align}

with bias-controlled Stark-mixing angle
\begin{equation}
\tan(2\theta)=\beta,\qquad
\beta\equiv \frac{2E_z|\mu_{34}^{z}|}{\hbar\Delta_{34}},\qquad
\Delta_{34}\equiv \omega_3-\omega_4>0,
\label{eq:app_theta_def}
\end{equation}
where $0\le \theta\le \pi/4$. The corresponding dressed splitting is
\begin{equation}
\omega_{34}^{(S)}\equiv \omega_3^{(S)}-\omega_4^{(S)}
=
\sqrt{(\omega_3-\omega_4)^2+\left(\frac{2E_z|\mu_{34}^{z}|}{\hbar}\right)^2}.
\label{eq:app_dressed_splitting}
\end{equation}

\subsection{Projected dipole operators in the Stark basis}

Introduce the polarization-projected dipole operators
\[
\hat{\mu}_p=\hat{\boldsymbol{\mu}}\!\cdot\!\mathbf{e}_p,\qquad
\hat{\mu}_c=\hat{\boldsymbol{\mu}}\!\cdot\!\mathbf{e}_c,\qquad
\hat{\mu}_S=\hat{\boldsymbol{\mu}}\!\cdot\!\mathbf{e}_S.
\]
In the minimal constructive model we retain only the bare matrix elements associated with the transitions $\ket{1}\leftrightarrow\ket{2}$, $\ket{2}\leftrightarrow\ket{3}$, and $\ket{3}\leftrightarrow\ket{4}$, and neglect all others:
\begin{align}
\hat{\mu}_p &\approx \mu_{12}^{p}\ket{1}\!\bra{2}+\mu_{21}^{p}\ket{2}\!\bra{1},\nonumber\\
\hat{\mu}_c &\approx \mu_{23}^{c}\ket{2}\!\bra{3}+\mu_{32}^{c}\ket{3}\!\bra{2},\nonumber\\
\hat{\mu}_S &\approx \mu_{34}^{S}\ket{3}\!\bra{4}+\mu_{43}^{S}\ket{4}\!\bra{3},
\label{eq:app_projected_dipoles}
\end{align}

Because the Stark mixing acts only in the $\{\ket{3},\ket{4}\}$ subspace, $\hat{\mu}_p$ is unchanged by the transformation.

Transforming the dipole operator $\hat{\mu}_c$ into the Stark basis gives
\begin{equation}
\hat{\mu}_c^{(S)}
=
\mu_{23}^{c}\cos\theta\,\ket{2}\!\bra{3^{(S)}}
-\mu_{23}^{c}\sin\theta\,\ket{2}\!\bra{4^{(S)}}
+\mathrm{H.c.}
\label{eq:app_mu_c_stark}
\end{equation}
The minus sign is a basis-convention artifact. By absorbing constant phases and signs into the definition of the Stark basis, the corresponding effective couplings in the main text may be taken real and nonnegative:
\[
\OmTwoThree=\Omega_c\cos\theta,\qquad
\OmTwoFour=\Omega_c\sin\theta.
\]

The transformed signal dipole operator $\hat{\mu}_S^{(S)}$ contains both off-diagonal and diagonal dressed-state terms. Choosing phases so that $\mu_{34}^{S}$ is real, one finds
\begin{equation}
\hat{\mu}_S^{(S)}
=
\mu_{34}^{S}\cos(2\theta)\Big(\ket{3^{(S)}}\!\bra{4^{(S)}}+\ket{4^{(S)}}\!\bra{3^{(S)}}\Big)
+
\mu_{34}^{S}\sin(2\theta)\Big(\ket{3^{(S)}}\!\bra{3^{(S)}}-\ket{4^{(S)}}\!\bra{4^{(S)}}\Big).
\label{eq:app_mu_s_stark}
\end{equation}
The off-diagonal term defines the effective Stark-basis coupling
\begin{equation}
\OmThreeFour=\Omega_S\cos(2\theta),
\label{eq:app_effective_signal_coupling}
\end{equation}
where
\[
\Omega_S=\frac{|\mu_{34}^{S}|A_S}{2\hbar}
\]
is the underlying signal Rabi scale associated directly with the received RF amplitude $A_S$. The diagonal term in~\eqref{eq:app_mu_s_stark} produces an oscillatory modulation of the Stark-state energies at $\omega_S$. In the minimal receiver model, this term is neglected as an additional simplifying approximation, and only the near-resonant off-diagonal coupling responsible for loop closure is retained. If retained, the diagonal modulation would add time-periodic diagonal Fourier blocks to the Hamiltonian but would not alter the structural phase-law result so long as the phase dependence still enters only through $(\omega_S t+\Phi_S)$.

\subsection{Rotating-wave approximation}

The time-dependent interaction Hamiltonian is
\begin{equation}
H_1(t)
=
-A_p\cos(\omega_p t)\hat{\mu}_p
-A_c\cos(\omega_c t)\hat{\mu}_c
-A_S\cos(\omega_S t+\Phi_S)\hat{\mu}_S.
\label{eq:app_H1t}
\end{equation}
After transforming into the Stark basis, discarding the diagonal modulation term in~\eqref{eq:app_mu_s_stark} as part of the minimal-model approximation, retaining the near-resonant off-diagonal couplings, and applying the RWA, one obtains
\begin{align}
H_{\mathrm{RWA}}(t)
&\approx
\hbar\Omega_p e^{-i\omega_p t}\ket{1}\!\bra{2}
+
\hbar\Omega_c\cos\theta\,e^{-i\omega_c t}\ket{2}\!\bra{3^{(S)}}
\nonumber\\
&\quad
+
\hbar\Omega_c\sin\theta\,e^{-i\omega_c t}\ket{2}\!\bra{4^{(S)}}
+
\hbar\Omega_S\cos(2\theta)\,e^{-i(\omega_S t+\Phi_S)}\ket{3^{(S)}}\!\bra{4^{(S)}}
+\mathrm{H.c.}
\label{eq:app_HRWA}
\end{align}
Here the minus sign multiplying the $\ket{2}\!\bra{4^{(S)}}$ term in~\eqref{eq:app_mu_c_stark} has been absorbed into the phase convention of $\ket{4^{(S)}}$, so the effective couplings $\OmTwoThree$ and $\OmTwoFour$ may be taken real and nonnegative.

\subsection{Co-rotating frame and gauge placement of the signal phase}

A convenient first rotating frame is generated by
\begin{equation}
U_0(t)=
\exp\!\Big[
-i\omega_p t\,\ket{2}\!\bra{2}
-i(\omega_p+\omega_c)t\,\ket{3^{(S)}}\!\bra{3^{(S)}}
-i(\omega_p+\omega_c)t\,\ket{4^{(S)}}\!\bra{4^{(S)}}
\Big].
\label{eq:app_U0}
\end{equation}
The corresponding Hamiltonian
\[
H_{\mathrm{rot}}^{(0)}=U_0(H_0'+H_{\mathrm{RWA}})U_0^\dagger+i\hbar\,\dot U_0 U_0^\dagger
\]
is, in angular-frequency units,
\begin{equation}
\frac{H_{\mathrm{rot}}^{(0)}}{\hbar}=
\begin{bmatrix}
0 & \Omega_p & 0 & 0\\
\Omega_p & -\Delta_p & \OmTwoThree & \OmTwoFour\\
0 & \OmTwoThree & -(\Delta_p+\Delta_c) & \OmThreeFour e^{-i(\omega_S t+\Phi_S)}\\
0 & \OmTwoFour & \OmThreeFour e^{i(\omega_S t+\Phi_S)} & -(\Delta_p+\Delta_c+\omega_{34}^{(S)})
\end{bmatrix},
\label{eq:app_Hrot0}
\end{equation}
with
\[
\Delta_p=\omega_p-(\omega_2-\omega_1),\qquad
\Delta_c=\omega_c-(\omega_3^{(S)}-\omega_2).
\]

In~\eqref{eq:app_Hrot0}, the received RF field is represented in the most direct way: it physically drives the Stark-state transition $\ket{3^{(S)}}\leftrightarrow\ket{4^{(S)}}$ and carries the explicit factor $e^{\pm i(\omega_S t+\Phi_S)}$. To obtain the form used in the main text, apply the additional diagonal transformation
\begin{equation}
W(t)=\exp\!\Big[-i(\omega_S t+\Phi_S)\ket{4^{(S)}}\!\bra{4^{(S)}}\Big].
\label{eq:app_W}
\end{equation}
Then
\[
\Heff = W\left(\frac{H_{\mathrm{rot}}^{(0)}}{\hbar}\right)W^\dagger + i\,\dot W W^\dagger,
\]
which gives
\begin{equation}
\Heff(t,\Phi_S)=
\begin{bmatrix}
0 & \Omega_p & 0 & 0\\
\Omega_p & -\Delta_p & \OmTwoThree & \OmTwoFour e^{i(\omega_S t+\Phi_S)}\\
0 & \OmTwoThree & -(\Delta_p+\Delta_c) & \OmThreeFour\\
0 & \OmTwoFour e^{-i(\omega_S t+\Phi_S)} & \OmThreeFour & -(\Delta_p+\Delta_c-\Delta_S)
\end{bmatrix},
\label{eq:app_H_standard}
\end{equation}
where
\begin{equation}
\Delta_S=\omega_S-\omega_{34}^{(S)}.
\label{eq:app_detuning_signal}
\end{equation}
Equation~\eqref{eq:app_H_standard} is the full matrix form of the reduced Hamiltonian used in the paper; it is equivalent to the decomposition
\[
\Heff(t,\Phi_S)=\Hzero+\Hplus e^{i(\omega_S t+\Phi_S)}+\Hminus e^{-i(\omega_S t+\Phi_S)}.
\]

The transformation~\eqref{eq:app_W} is a gauge choice. It does not reassign which field physically drives which transition; it only relocates the explicit phase factor within the loop. In the representation~\eqref{eq:app_H_standard}, the received-signal amplitude enters through the static coupling $\OmThreeFour$, while the received-signal phase enters only through the combination $(\omega_S t+\Phi_S)$. The physical Hamiltonian corresponding to~\eqref{eq:app_H_standard} is $\hbar\Heff$.

\section{Floquet--Liouville formulation and block-tridiagonal linear system}
\label{app:floquet_linear_system}

This appendix records a convenient Liouville-space form of the truncated harmonic-balance problem. Because of the exact harmonic phase law proved in Section~\ref{sec:pss}, only the reference-phase coefficients
\[
P^{(n)}=\rho^{(n)}(0)
\]
need to be solved numerically.

\subsection{Vectorization identities and Liouvillian blocks}

For a $4\times4$ matrix $X$, define $\mathrm{vec}(X)$ by column stacking. Then
\begin{equation}
\mathrm{vec}(AXB)=(B^{\mathrm{T}}\otimes A)\,\mathrm{vec}(X),
\label{eq:app_vec_identity}
\end{equation}
which implies
\begin{align}
\mathrm{vec}([H,X]) &= (I_4\otimes H-H^{\mathrm{T}}\otimes I_4)\,\mathrm{vec}(X), \label{eq:app_comm_vec}\\
\mathrm{vec}(LXL^\dagger) &= (L^*\otimes L)\,\mathrm{vec}(X), \label{eq:app_jump_vec}\\
\mathrm{vec}(\{G,X\}) &= (I_4\otimes G+G^{\mathrm{T}}\otimes I_4)\,\mathrm{vec}(X). \label{eq:app_anti_vec}
\end{align}

Let
\[
p^{(n)}\triangleq \mathrm{vec}\!\left(P^{(n)}\right)\in\mathbb{C}^{16}.
\]
The dissipator in vectorized form is
\begin{equation}
\mathbb{D}
=
\sum_k
\left[
L_k^*\otimes L_k
-\frac{1}{2}\Big(
I_4\otimes L_k^\dagger L_k
+
(L_k^\dagger L_k)^{\mathrm{T}}\otimes I_4
\Big)
\right].
\label{eq:app_D_super}
\end{equation}
Define the three harmonic Liouvillian blocks
\begin{align}
\mathbb{L}_0
&=
-i\Big(I_4\otimes \Hzero-\Hzero^{\mathrm{T}}\otimes I_4\Big)+\mathbb{D},
\label{eq:app_L0_super}
\\
\mathbb{L}_+
&=
-i\Big(I_4\otimes \Hplus-\Hplus^{\mathrm{T}}\otimes I_4\Big),
\label{eq:app_Lplus_super}
\\
\mathbb{L}_-
&=
-i\Big(I_4\otimes \Hminus-\Hminus^{\mathrm{T}}\otimes I_4\Big).
\label{eq:app_Lminus_super}
\end{align}
Then the harmonic-balance equations~\eqref{eq:harmonic_balance_eq} become
\begin{equation}
\Big(in\omega_S I_{16}-\mathbb{L}_0\Big)p^{(n)}
-\mathbb{L}_+\,p^{(n-1)}
-\mathbb{L}_-\,p^{(n+1)}
=0,
\qquad n\in\mathbb{Z}.
\label{eq:app_harmonic_balance_vec}
\end{equation}
The nearest-neighbor structure in harmonic index follows directly from the fact that the Hamiltonian contains only Fourier components at $0$ and $\pm1$.

\subsection{Global truncated system and trace constraints}

Truncating the harmonics to $n\in[-N,N]$ gives the stacked unknown
\begin{equation}
\vec{\mathcal{P}}
=
\Big(
p^{(-N)},\,p^{(-N+1)},\,\ldots,\,p^{(N)}
\Big)^{\mathrm{T}}
\in \mathbb{C}^{16(2N+1)}.
\label{eq:app_global_stack}
\end{equation}
With
\[
\mathbb{A}_n \triangleq in\omega_S I_{16}-\mathbb{L}_0,\qquad
\mathbb{B}_- \triangleq -\mathbb{L}_+,\qquad
\mathbb{B}_+ \triangleq -\mathbb{L}_-,
\]
the truncated system has the block-tridiagonal form
\begin{equation}
\begin{bmatrix}
\mathbb{A}_{-N} & \mathbb{B}_+ & 0 & \cdots & 0\\
\mathbb{B}_- & \mathbb{A}_{-N+1} & \mathbb{B}_+ & \ddots & \vdots\\
0 & \ddots & \ddots & \ddots & 0\\
\vdots & \ddots & \mathbb{B}_- & \mathbb{A}_{N-1} & \mathbb{B}_+\\
0 & \cdots & 0 & \mathbb{B}_- & \mathbb{A}_{N}
\end{bmatrix}
\vec{\mathcal{P}}
=
0.
\label{eq:app_block_tridiag}
\end{equation}

The physical trace constraints are
\begin{equation}
\mathrm{Tr}\,P^{(0)}=1,\qquad
\mathrm{Tr}\,P^{(n)}=0\quad (n\neq 0).
\label{eq:app_trace_constraints}
\end{equation}
If
\[
\tau^{\mathrm{T}}\triangleq \mathrm{vec}(I_4)^{\mathrm{T}},
\]
then these constraints become
\[
\tau^{\mathrm{T}}p^{(0)}=1,\qquad
\tau^{\mathrm{T}}p^{(n)}=0\quad (n\neq 0).
\]
In practice, one redundant equation per harmonic block is replaced by the corresponding trace constraint. Solving the constrained truncated system yields the reference-phase coefficients $P^{(n)}$, from which the coefficients at arbitrary $\Phi_S$ follow immediately via the exact harmonic phase law~\eqref{eq:factorization}. Convergence with harmonic truncation is checked a posteriori; numerical evidence for the production choice $N=3$ is given in Appendix~\ref{app:numerics_validation}.

\section{Perturbative scaling and initialization rule}
\label{app:perturbative_scaling}

This appendix derives the weak-coupling angular scaling that motivates the perturbative initialization rule for the Stark-mixing angle. The result is an initialization heuristic for the full Floquet--Liouville search, not a universal optimum theorem.

\subsection{Leading-order scaling of the first harmonic}

In the weak-signal, weak-loop-closure regime, the first-harmonic probe response is generated by the smallest process that couples the $n=0$ and $n=\pm1$ harmonic sectors while returning to the probe coherence. To lowest nonvanishing order, this requires one probe interaction, two upper-leg optical interactions, and one signal-mediated Stark-state coupling. Consequently,
\begin{equation}
|P_{21}^{(1)}|
\propto
\Omega_p\,\OmTwoThree\,\OmTwoFour\,\OmThreeFour
\times
\mathcal{F}(\Delta_p,\Delta_c,\{\gamma\}),
\label{eq:app_perturbative_general}
\end{equation}
where $\mathcal{F}$ collects the detuning and linewidth dependence.

Substituting the constructive-model couplings
\[
\OmTwoThree=\Omega_c\cos\theta,\qquad
\OmTwoFour=\Omega_c\sin\theta,\qquad
\OmThreeFour=\Omega_S\cos(2\theta)
\]
gives
\begin{equation}
|P_{21}^{(1)}|
\propto
\Omega_p\,\Omega_c^2\,\Omega_S\,
f(\theta)\,
\mathcal{F}(\Delta_p,\Delta_c,\{\gamma\}),
\qquad
f(\theta)=\sin\theta\,\cos\theta\,\cos(2\theta).
\label{eq:app_perturbative_theta}
\end{equation}
Thus, to leading order, the first harmonic is linear in the underlying signal scale $\Omega_S$ and inherits all of its angular dependence through $f(\theta)$.

\subsection{Consequences for the phase and amplitude design metrics}

In the notation of Section~\ref{subsec:three_optima}, the perturbative form~\eqref{eq:app_perturbative_theta} implies
\begin{equation}
m(\Omega_S,\theta)\approx C\,\Omega_S\,f(\theta),
\label{eq:app_m_linear}
\end{equation}
for a constant $C$ that depends on the optical settings, detunings, and linewidths but not on $\theta$. Hence the phase metric obeys
\begin{equation}
M_{\phi}(\theta)=m(\Omega_{S,0},\theta)\approx C\,\Omega_{S,0}\,f(\theta),
\label{eq:app_Mphi_pert}
\end{equation}
while the amplitude metric becomes
\begin{equation}
M_A(\theta)
=
\Omega_{S,0}\left|\frac{\partial m}{\partial\Omega_S}(\Omega_{S,0},\theta)\right|
\approx
C\,\Omega_{S,0}\,|f(\theta)|.
\label{eq:app_MA_pert}
\end{equation}
Therefore, in the perturbative regime, the phase and amplitude criteria inherit the same angular dependence. The fact that the full model later yields distinct optima $\theta_{\phi}^{\star}\neq\theta_A^{\star}$ is consequently a genuinely non-perturbative effect.

\subsection{Perturbative initialization rule}

Using the identity
\begin{equation}
f(\theta)=\frac{1}{4}\sin(4\theta)=\frac{\beta}{2(1+\beta^2)},
\qquad
\beta=\tan(2\theta),
\label{eq:app_f_theta}
\end{equation}
one finds that $f(\theta)$ is maximized on $\theta\in[0,\pi/4]$ at
\begin{equation}
\theta_{\mathrm{seed}}=\frac{\pi}{8},
\qquad
\beta_{\mathrm{seed}}=1.
\label{eq:app_theta_seed}
\end{equation}
Using the definition of $\beta$ in~\eqref{eq:app_theta_def} then gives the corresponding seed bias amplitude
\begin{equation}
E_{z,\mathrm{seed}}=\frac{\hbar\Delta_{34}}{2|\mu_{34}^{z}|}.
\label{eq:app_Eseed}
\end{equation}
This rule should be used only to initialize the full-model search over $\theta$. Once non-perturbative saturation, detuning dependence, and linewidth effects become relevant, the practically meaningful operating points are the full-model optima $\theta_{\phi}^{\star}$, $\theta_{A}^{\star}$, and, after application weights are fixed, $\theta_{J}^{\star}$.

\section{High-SNR error laws and calibration assumptions}
\label{app:rmse_derivation}

This appendix derives the leading-order error laws used in Section~\ref{sec:performance}. The derivation assumes a calibrated harmonic-phasor model: the dominant randomness is additive complex noise on the measured demodulated harmonic, while the reference harmonic used for phase calibration and the response map used for amplitude inversion are treated as exact or as sufficiently high-SNR calibration quantities that their uncertainty is negligible at leading order.

\subsection{Calibrated phasor model}

Let
\begin{equation}
z_n=\rho_{21}^{(n)}(\Phi_S)=r_n e^{i\varphi_n},
\qquad
r_n=|P_{21}^{(n)}|,
\qquad
\varphi_n=\arg(P_{21}^{(n)})+n\Phi_S.
\label{eq:app_phasor_model}
\end{equation}
The calibrated measured harmonic is modeled as
\begin{equation}
\widehat{z}_n=z_n+\eta_n,
\qquad
\eta_n=\eta_{I,n}+i\eta_{Q,n},
\label{eq:app_noisy_phasor}
\end{equation}
with
\[
\eta_{I,n},\eta_{Q,n}\sim\mathcal{N}(0,\sigma_n^2)
\]
independent and zero mean. The harmonic SNR is therefore
\begin{equation}
\SNR_n=\frac{r_n^2}{2\sigma_n^2}.
\label{eq:app_snr}
\end{equation}

\subsection{Phase error}

For the phase estimator, the reference ratio in~\eqref{eq:phase_estimator} removes fixed readout gain and fixed phase offsets. If the reference harmonic is treated as exact, the local phase error is the phase perturbation of $\widehat{z}_n$ itself. Rotate the coordinate system so that $z_n=r_n$ lies on the positive real axis. Then
\[
\widehat{z}_n=r_n+\eta_{I,n}+i\eta_{Q,n}.
\]
At high SNR,
\[
\arg(\widehat{z}_n)\approx \frac{\eta_{Q,n}}{r_n},
\]
so the phase-estimator error is
\[
\delta\widehat{\Phi}_S\approx \frac{1}{n}\frac{\eta_{Q,n}}{r_n}.
\]
Hence
\begin{equation}
\mathrm{Var}(\widehat{\Phi}_S)\approx \frac{\sigma_n^2}{n^2r_n^2}
=\frac{1}{2n^2\SNR_n},
\qquad
\RMSE_{\Phi}\approx \frac{1}{n\sqrt{2\SNR_n}}.
\label{eq:app_phase_rmse}
\end{equation}

If the reference harmonic is itself noisy and statistically independent, its phase uncertainty adds at leading order to the variance above. More precisely, an independent reference phasor with magnitude $r_{n,\mathrm{ref}}$ and quadrature variance $\sigma_{n,\mathrm{ref}}^2$ contributes an additional term of order $\sigma_{n,\mathrm{ref}}^2/(n^2 r_{n,\mathrm{ref}}^2)$.

\subsection{Amplitude error}

For amplitude recovery, let the true signal strength be $\Omega_S$ and write the calibrated response map as
\[
m_n(\Omega_S)=|P_{21}^{(n)}(\Omega_S)|.
\]
At high SNR, the noisy magnitude satisfies
\[
|\widehat{z}_n|=|z_n+\eta_n|\approx r_n+\eta_{\parallel,n},
\]
where $\eta_{\parallel,n}$ is the noise component parallel to $z_n$, with variance $\sigma_n^2$. Linearizing the inverse response map gives
\[
\widehat{\Omega}_S-\Omega_S
\approx
\frac{\eta_{\parallel,n}}{m_n'(\Omega_S)}.
\]
Therefore
\begin{equation}
\mathrm{Var}(\widehat{\Omega}_S)\approx \frac{\sigma_n^2}{|m_n'(\Omega_S)|^2}.
\label{eq:app_amp_var}
\end{equation}
Using the logarithmic sensitivity
\[
s_n(\Omega_S)=\frac{\Omega_S}{m_n(\Omega_S)}m_n'(\Omega_S)
\]
and $m_n(\Omega_S)=r_n$ yields
\[
|m_n'(\Omega_S)|=\frac{|s_n(\Omega_S)|\,r_n}{\Omega_S}.
\]
Substituting into~\eqref{eq:app_amp_var} and using~\eqref{eq:app_snr} gives
\begin{equation}
\frac{\RMSE_{\Omega_S}}{\Omega_S}
\approx
\frac{1}{|s_n(\Omega_S)|\sqrt{2\SNR_n}}.
\label{eq:app_amp_rmse}
\end{equation}
It is convenient to define the effective amplitude SNR
\begin{equation}
\SNR_{A,n}^{\mathrm{eff}}(\Omega_S)\triangleq |s_n(\Omega_S)|^2\,\SNR_n,
\label{eq:app_amp_snr_eff}
\end{equation}
so that
\begin{equation}
\frac{\RMSE_{\Omega_S}}{\Omega_S}
\approx
\frac{1}{\sqrt{2\,\SNR_{A,n}^{\mathrm{eff}}(\Omega_S)}}.
\label{eq:app_amp_rmse_eff}
\end{equation}

\subsection{Extension to nonuniform bias and validity range}

The same local linearization applies to the nonuniform-bias case after replacing the uniform quantities $(m_n,s_n,\SNR_n)$ by the averaged quantities $(m_{\mathrm{avg}},s_{\mathrm{avg}},\SNR_{1,\mathrm{eff}})$ introduced in Section~\ref{sec:bias_cases}. This gives the averaged amplitude law used in~\eqref{eq:amp_rmse_nonuniform} and the corresponding effective amplitude SNR in~\eqref{eq:snr_A_eff_theory}.

All formulas in this appendix are local high-SNR approximations. They neglect phase wrapping, finite-branch ambiguity of the inverse response map, and higher-order curvature of the phase and amplitude estimators over the error scale of interest. Deviations from the asymptotic laws at low SNR are therefore expected and are observed numerically in Section~\ref{sec:numerics}.

\section{Numerical implementation, convergence, and validation}
\label{app:numerics_validation}

This appendix collects numerical checks and implementation details omitted from the main text. Unless otherwise noted, all frequencies are expressed in units of $\gamma_{21}$, which is set to unity.

\subsection{Harmonic truncation and convergence}

The production calculations truncate the harmonic-balance system to $n\in[-N,N]$ with $N=3$. To assess the truncation error, we compare the first-harmonic phasor computed at order $N$ with a higher-order reference solution at $N_{\mathrm{ref}}=8$ via
\begin{equation}
\varepsilon_N
\triangleq
\frac{\left|P_{21,N}^{(1)}-P_{21,N_{\mathrm{ref}}}^{(1)}\right|}
{\left|P_{21,N_{\mathrm{ref}}}^{(1)}\right|}.
\label{eq:app_rel_error}
\end{equation}
Two operating points are used in the convergence study:

\begin{itemize}
\item the nominal point from~\eqref{eq:nominal_point},
\item a more demanding stress point
\[
(\Omega_c,\Omega_S,\Delta_p,\Delta_c,\theta,\omega_S)
=
(1.6,0.25,0.10,-0.10,0.60,1.0).
\]
\end{itemize}

At the nominal point, $\varepsilon_N$ already falls below $10^{-15}$ by $N=2$. At the stress point, it drops below $10^{-7}$ by $N=3$, below $10^{-11}$ by $N=4$, and reaches machine precision by $N=5$. These checks justify the production choice $N=3$ for the parameter range studied in the main text.

\subsection{Time-domain validation and Monte-Carlo protocol}

As a second validation step, the periodic steady state reconstructed from the truncated Floquet coefficients is compared with direct time-domain integration of the master equation using a stiff BDF solver. The time-domain solver is run for a burn-in of $180$ received-signal periods, followed by an evaluation window of $6$ periods sampled at $400$ points per period. Over the tested parameter range, the direct integration and the Floquet reconstruction are visually indistinguishable in the periodic regime.

Figure~\ref{fig:app_validation} summarizes both validation checks.

\begin{figure}[t]
\centering
\includegraphics[width=0.92\linewidth]{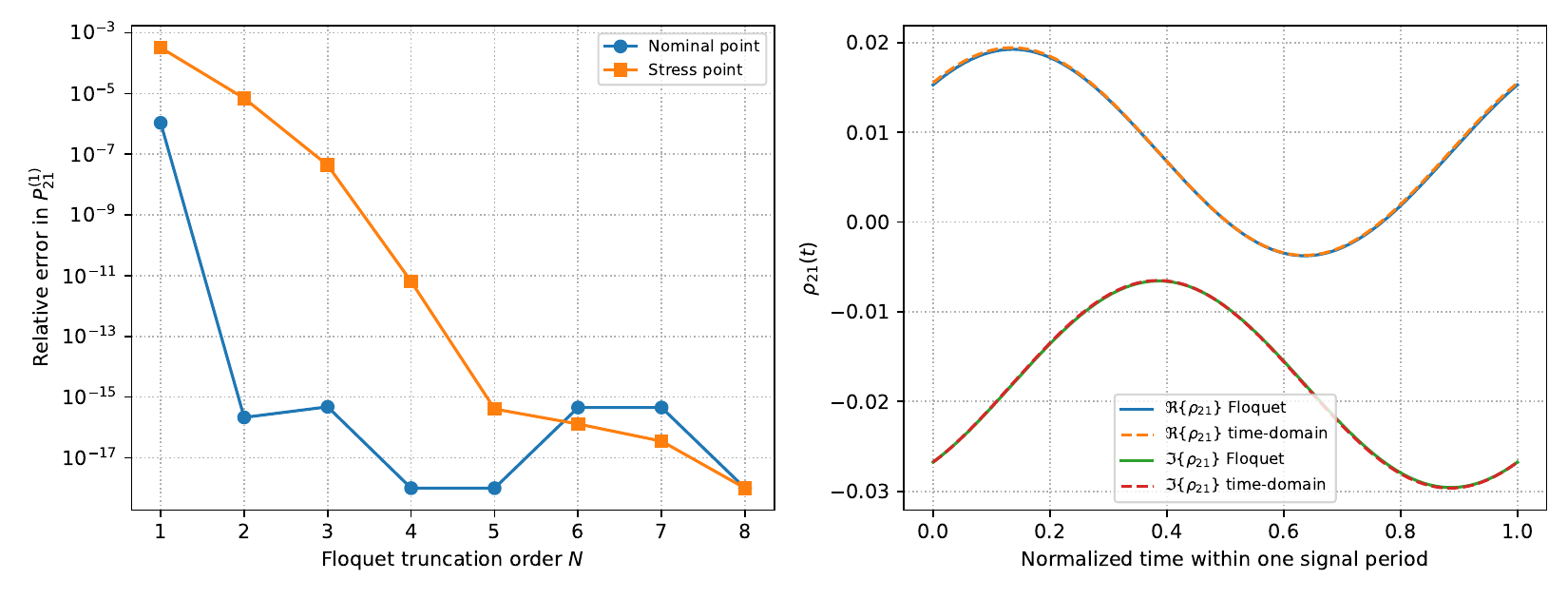}
\caption{Numerical validation of the Floquet--Liouville solver. Panel (a) shows the relative error in the first-harmonic phasor versus harmonic truncation order $N$, referenced to the higher-order solution with $N_{\mathrm{ref}}=8$, for the nominal point~\eqref{eq:nominal_point} and the stress point specified in the text. Panel (b) compares direct time-domain integration of the master equation with Floquet reconstruction of $\rho_{21}(t)$ over one received-signal period.}
\label{fig:app_validation}
\end{figure}

The Monte-Carlo RMSE curves reported in Section~\ref{sec:numerics} use $3\times 10^4$ trials per SNR point. For each trial, the signal phase is drawn uniformly from $[0,2\pi)$, circular complex Gaussian noise is added to the calibrated harmonic phasor, and the phase and amplitude estimators are applied exactly as described in Section~\ref{sec:estimators}.

\subsection{Parameterization of nonuniform bias}

For the nonuniform-bias calculations, the local bias is parameterized through
\[
\beta=\tan(2\theta),
\]
which is monotone in the Stark-mixing angle on $\theta\in[0,\pi/4]$. The numerical generator uses the relative spread $\sigma_\beta/\beta_0$ as the dimensionless measure of nonuniformity. For each chosen spread level, local first-harmonic responses $P_{21}^{(1)}(\beta)$ are computed and combined with normalized weights to approximate the spatial average in~\eqref{eq:averaged_harmonic}. The resulting coherent-gain factor is
\[
G=\frac{|\overline{P}_{21}^{(1)}|}{|P_{21}^{(1)}(\beta_0)|}.
\]

Figure~\ref{fig:app_G_vs_sigma} shows the coherent-gain factor versus the relative nonuniformity level. In the parameter regime used for the main-text collapse study, the representative values are
\[
G\approx 0.32,\ 0.22,\ 0.10
\qquad\text{for}\qquad
\sigma_\beta/\beta_0=0.01,\ 0.02,\ 0.05,
\]
respectively. The main text tests the nonuniform-bias theory through the RMSE collapse laws, while Fig.~\ref{fig:app_G_vs_sigma} provides a supporting diagnostic for the magnitude of the coherent averaging loss itself.

\begin{figure}[t]
\centering
\includegraphics[width=0.85\linewidth]{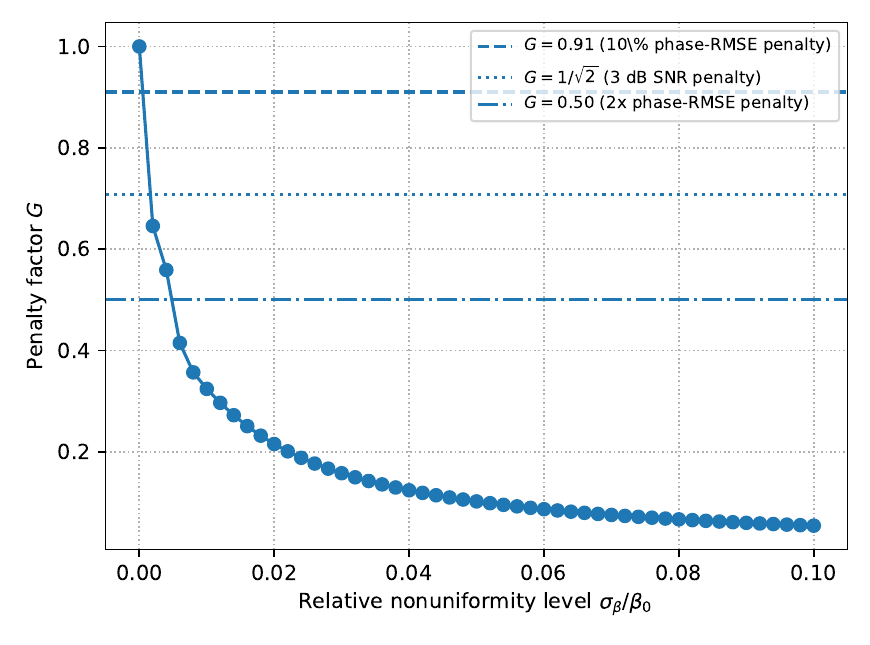}
\caption{Supporting diagnostic for nonuniform-bias calculations: coherent-gain factor $G$ versus the relative nonuniformity level $\sigma_{\beta}/\beta_0$. The horizontal guides translate $G$ into representative phase-estimation penalties relative to the corresponding uniform-bias benchmark, including a $10\%$ phase-RMSE increase, a $3$ dB phase-SNR reduction, and a twofold phase-RMSE increase.}
\label{fig:app_G_vs_sigma}
\end{figure}

\bibliographystyle{IEEEtran}
\bibliography{citations}

\end{document}